\PassOptionsToPackage{unicode,final}{hyperref}
\PassOptionsToPackage{hyphens}{url}
\documentclass[pdflatex,sn-mathphys-num]{sn-jnl}%

\usepackage{amsmath,amssymb}
 \usepackage{amsthm}

\newtheorem{thm}{Theorem}

\newtheorem{cor}[thm]{Corollary}
\makeatletter
\def\fps@figure{htbp}
\makeatother

\usepackage{graphicx}
\usepackage{booktabs}
\usepackage{float}
\usepackage{array}
\newcommand{\figlink}[2]{\hyperref[#1]{\ref*{#1}#2}}

\graphicspath{{fig/}{src/fig/}{../open_repo/agent_verify/results_si/exp1_eigsel/}{../open_repo/agent_verify/results_si/exp2_dprime/}{../open_repo/agent_verify/results_si/exp3_grid/}{open_repo/agent_verify/results_si/exp1_eigsel/}{open_repo/agent_verify/results_si/exp2_dprime/}{open_repo/agent_verify/results_si/exp3_grid/}}

\begin{document}

\title{
        Inducing fast-slow separation enables robust learning of complex dynamics
}
\vspace{1cm}

\author*[1]{\fnm{Satoshi} \sur{Oishi}}
\author[2,1]{\fnm{Hiroshi} \sur{Yamashita}}
\author[1,2]{\fnm{Hideyuki} \sur{Suzuki}}
\author[1,2]{\fnm{Sho} \sur{Shirasaka}}

\affil*[1]{\orgdiv{Graduate School of Information Science and Technology}, \orgname{\\The University of Osaka}, \orgaddress{\street{1-5 Yamadaoka}, \city{Suita}, \postcode{565-0871}, \state{Osaka}, \country{Japan}}}
\affil*[2]{\orgdiv{International Research Center for Neurointelligence}, \orgname{\\The University of Tokyo}, \orgaddress{\street{7-3-1 Hongo}, \city{Bunkyo-ku}, \postcode{113-0033}, \state{Tokyo}, \country{Japan}}}
\date{\today}

\abstract{
Many real-world systems share a fast-slow structure: most degrees of freedom relax quickly, while a few slow modes govern the long-term evolution. Their dynamics often collapse onto low-dimensional complex structures such as chaotic attractors, and a goal of nonlinear science is to predict and faithfully reproduce them from observed time series. Machine-learning models and data-driven approaches can embed a chaotic attractor in high-dimensional spaces, but embedding alone does not guarantee faithful reproduction. Training can create spurious slow modes, excess modes unnecessary for the target dynamics, which destabilize the reconstruction. To prevent this instability, we introduce input-layer designs for reservoir computing, a framework suited to physical implementation. Through restriction of network controllability, the designs limit the number of slow modes available to learning and anchor the remaining modes to stay fast in advance, even for a black-box model. The reconstruction is thereby confined to a transversally attracting subspace, and spurious slow modes are suppressed. Across diverse chaotic systems, the designs robustly reproduce the attractors and their dynamical invariants, extend the horizon of accurate prediction, and withstand perturbations of the internal weights, without extensive tuning. Inducing fast-slow separation and reconstructing attractors in attracting subspaces offers a design principle for reliable data-driven modeling.
}

\maketitle

\newpage

\section{Main}\label{main}

Many complex systems throughout nature share a fast-slow structure. Most of their degrees of freedom relax quickly, while a few slow modes govern the long-term evolution. The dynamics therefore collapse onto an emergent low-dimensional structure, as in the atmosphere \cite{Ghil2020-jw} and neural populations \cite{Vyas2020-ie}.
However, the asymptotic dynamics can be highly complex, a behavior known as chaos. A chaotic attractor draws every nearby trajectory onto itself, yet the motion on it is unpredictable because tiny differences in the initial state grow exponentially.
A long-standing goal of data-driven science is to predict and faithfully reproduce such chaotic dynamics from observed time series alone.

A standard route to modeling such systems from data is attractor reconstruction: lifting the observed time series into a higher-dimensional space (e.g., by delay-coordinate embedding \cite{Takens1981-mw, Sauer1991-ew}) recovers a set topologically equivalent to the original attractor. From the reconstructed set, invariant properties of the dynamics can be estimated and its future evolution predicted. Machine-learning models can also reconstruct the attractor, by mapping the input data into their internal representation spaces.
Recurrent neural networks (RNNs), for instance, can recover the topology of the attractor within their internal representation, but when the target is chaotic, their training suffers from difficulties rooted in the sensitivity to initial conditions, most notably the exploding-gradient problem \cite{Mikhaeil2021-ve}.
Reservoir computing (RC) \cite{Jaeger2001-vs, Maass2002-me} circumvents these difficulties. It employs a fixed, randomly generated internal network (the reservoir) to project inputs into a high-dimensional state space, and trains only the readout layer via linear regression, so that no gradient propagates through the recurrent dynamics (Fig.~\figlink{fig:fig1}{a}). This simplicity also allows the reservoir to be realized directly in physical substrates---photonic, spintronic, and quantum \cite{Tanaka2019-vn}---where the internal network is fixed by the hardware rather than designed. RC has proven effective in modeling chaotic time series \cite{Pathak2017-id, Lu2018-vk}, a success attributed to the attractor being reconstructed within the reservoir state space \cite{Hart2020-kx, Grigoryeva2021-mg}. The geometry of the target is thereby preserved, enabling both short-term prediction and the reproduction of its invariant properties \cite{Pathak2017-id, Margazoglou2023-lc, Sisodia2024-fl}.

However, the topological/geometrical reconstruction of the attractor alone does not guarantee the faithful reproduction of the dynamical system, especially during closed-loop (autonomous) operation (Fig.~\figlink{fig:fig1}{a}).
This limitation is not specific to RC but applies to any approach that predicts recursively \cite{Bi2023-jr, Kochkov2024-ak}. Such a model can stray from the reconstructed attractor, because training does not control the transverse stability of the reconstruction within the representation space. When the reconstruction is transversally unstable, trajectories diverge from it, leading to behavior entirely different from that of the original dynamical system (Fig.~\figlink{fig:fig1}{d}). Stabilizing the reconstructed attractor is therefore the central requirement, yet no method guarantees it in advance. In practice, RC thus often requires trial-and-error tuning \cite{Platt2023-lh}.

Here, we propose a design principle that offers an \textit{a priori} solution to stabilize this reconstruction.
We observed that such divergence is caused by \emph{spurious slow modes}: training perturbs the time scales of the network indiscriminately, creating excess slow modes that are unnecessary for reproducing the target dynamics (Fig.~\figlink{fig:fig1}{c,d}).
To suppress these modes, we developed an input-layer design that restricts the number of time scales available to learning and anchors all the remaining eigenvalues at fast time scales, thereby bounding the controllability of the network (Fig.~\figlink{fig:fig1}{e}).
This induces a fast-slow separation in the closed-loop dynamics and places the reconstruction near a low-dimensional, transversally attracting subspace designed before learning (Fig.~\figlink{fig:fig1}{b,f}).
We call this design goal \emph{attractor reconstruction in attracting subspaces} (ARAS).
The principle is not specific to RC and should benefit other data-driven models of dynamics.

In this work, we demonstrate the effectiveness of this principle in RC, a particularly simple architecture.
Across a wide range of chaotic systems, the design delivers long and robust predictions and reproduces the invariant properties of the target dynamics.
We further extended the design to reservoirs that are black boxes, as in physical implementations of RC---the \emph{probe-based construction}---with no loss of stability or accuracy.
We confirmed that the resulting models are robust not only to observational and internal noise but also to small changes in the internal coupling weights.

\begin{figure}[H]
    \centering
    \includegraphics[width=0.99\linewidth,height=\textheight,keepaspectratio]{figures/fig1.pdf}
    \caption{\textbf{a}, Schematic of reservoir computing (RC). During training (open loop), the RC is driven by the observed time series \(\mathbf{u}_t\) and only the readout \(\mathbf{W}_{\mathrm{out}}\) is learned. For autonomous prediction, the RC feeds its own one-step-ahead prediction back as the input, closing the loop into the autonomous dynamics governed by \(\mathbf{W}_{\mathrm{cl}}=\mathbf{A}+\mathbf{W}_{\mathrm{in}}\mathbf{W}_{\mathrm{out}}\).
    \textbf{b}, The design goal of this study. The reconstructed data are placed near a low-dimensional attracting slow submanifold of the representation space, and the fast contraction transverse to it renders the reconstruction transversally stable.
    \textbf{c},\textbf{d}, Instability in the conventional (random) input-layer design.
    \textbf{c}, Eigenvalue spectrum of the closed-loop matrix \(\mathbf{W}_{\mathrm{cl}}\), showing the uncontrolled emergence of spurious slow modes: eigenvalues escape beyond the spectral radius of \(\mathbf{A}\) (green dashed circle).
    \textbf{d}, Long-term autonomous prediction trajectory, which diverges from the true attractor and exhibits spurious behavior unfaithful to the target dynamics.
    \textbf{e},\textbf{f}, The proposed input-layer design.
    \textbf{e}, Eigenvalue spectrum of \(\mathbf{W}_{\mathrm{cl}}\). By design, only the three eigenvalues corresponding to the learnable modes are slow, while all remaining eigenvalues stay fast and within the stable region defined by the spectral radius of \(\mathbf{A}\), suppressing spurious slow modes.
    \textbf{f}, Long-term autonomous prediction trajectory, accurately and stably reproducing the target attractor.}\label{fig:fig1}
\end{figure}

\section{Results}\label{results}

\subsection{Suppression of spurious slow modes stabilizes attractor reconstruction}\label{suppression-of-spurious-slow-modes-stabilizes-attractor-reconstruction}

Following the standard RC procedure for predicting and reproducing dynamical systems, we feed the observed time series \(\mathbf{u}_t \in \mathbb{R}^{D}\) of the target system into the reservoir (open loop; Fig.~\figlink{fig:fig1}{a}):
\begin{equation}\label{eq:open_loop}
    \mathbf{r}_{t+1} = \boldsymbol{\tau}(\mathbf{A} \mathbf{r}_t + \mathbf{W}_{\mathrm{in}} \mathbf{u}_t + \mathbf{b}),
\end{equation}
where \(\mathbf{r}_t \in \mathbb{R}^{N}\) is the reservoir state, the transition matrix \(\mathbf{A}\), the input weights \(\mathbf{W}_{\mathrm{in}}\), and the bias \(\mathbf{b}\) are fixed, and \(\boldsymbol{\tau} = \boldsymbol{\tanh}\) acts as the element-wise activation function (Methods~\ref{reservoir-computing-fundamentals-and-closed-loop-operation}).
Training determines the readout \(\mathbf{W}_{\mathrm{out}}\) by linear regression.
Feeding the prediction \(\hat{\mathbf{u}}_t = \mathbf{W}_{\mathrm{out}}\mathbf{r}_t\) back as the input then closes the loop:
\begin{equation}\label{eq:closed_loop_main}
    \mathbf{r}_{t+1} = \boldsymbol{\tau}(\mathbf{W}_{\mathrm{cl}}\mathbf{r}_t + \mathbf{b}), \qquad \mathbf{W}_{\mathrm{cl}} := \mathbf{A} + \mathbf{W}_{\mathrm{in}}\mathbf{W}_{\mathrm{out}}.
\end{equation}
The closed-loop RC is an autonomous dynamical system whose time scales depend on the spectrum of the closed-loop matrix \(\mathbf{W}_{\mathrm{cl}}\).

Before training, all eigenvalues of \(\mathbf{A}\) lie within the small spectral radius \(\rho < 1\) set by design, so every mode of the open-loop reservoir is fast and contracting.
Training the readout adds the perturbation \(\mathbf{W}_{\mathrm{in}}\mathbf{W}_{\mathrm{out}}\) to \(\mathbf{A}\) and moves the eigenvalues.
In standard RC, where \(\mathbf{A}\) and \(\mathbf{W}_{\mathrm{in}}\) are both random, nothing constrains this movement.
For a standard RC trained on the Lorenz system \cite{Lorenz1963-kf} (\(D=3\), \(N=200\)), 23 eigenvalues of \(\mathbf{W}_{\mathrm{cl}}\) lie outside the spectral radius of \(\mathbf{A}\) in the realization shown (green dashed circle, Fig.~\figlink{fig:bs}{a}), although reproducing the three-dimensional target dynamics should require only three essential modes.
These excess slow eigenvalues give rise to the spurious slow modes (Fig.~\figlink{fig:fig1}{d}).
Both proposed constructions instead restrict this movement \textit{a priori}: the input layer is built on a \(D'\)-dimensional invariant subspace \(V_{D'}\) of the reservoir's linearized dynamics, and exactly \(D'=3\) eigenvalues leave the spectrum of \(\mathbf{A}\) while the remaining \(N-D'\) are preserved (Figs.~\figlink{fig:bs}{b,c})---exactly for the analytic construction (Theorem~\ref{thm:invariant_eigenvalues}; Methods~\ref{deterministic-input-layer-construction-proposed-method}) and approximately for the probe-based one ; Methods~\ref{probe-based-construction}.

The proposed constructions also aim to reconstruct the attractor within a transversally attracting subspace of the reservoir state space (Fig.~\figlink{fig:fig1}{b}).
By design, \(V_{D'}\) remains invariant under the closed-loop matrix \(\mathbf{W}_{\mathrm{cl}}\), and the \(D'\) moved eigenvalues are those associated with it (Corollary~\ref{cor:inheritance}).
In the directions transverse to \(V_{D'}\), the linearized closed loop contracts at the fast rates of the \(N-D'\) anchored eigenvalues, so \(V_{D'}\) remains attracting even after training.
It is thus desirable that the attractor be reconstructed within this attracting subspace, and the proposed constructions indeed achieve this (Figs.~\figlink{fig:bs}{e,f}).
Measuring the distance from each state of the reconstruction to its input subspace \(\mathbf{r}^{*}+\mathrm{Im}(\mathbf{W}_{\mathrm{in}})\), where \(\mathbf{r}^{*}\) is the input-free equilibrium of the reservoir, confirms this placement: both proposed constructions hold the reconstruction within a thin neighborhood of it (mean distance \(\approx 0.003\)--\(0.006\); Figs.~\figlink{fig:bs}{e,f}), whereas the random reconstruction is delocalized (\(\approx 0.25\), roughly forty times larger; Fig.~\figlink{fig:bs}{d}, and see also Extended Data Fig.~\ref{fig:ed-sweep-dist}).
The proposed constructions attain this confinement because their \(\mathrm{Im}(\mathbf{W}_{\mathrm{in}})\) is \(\mathbf{A}\)-invariant, so in the open loop (Eq.~\eqref{eq:open_loop}) both the drive and the memory of past inputs keep the state inside this subspace, and only the nonlinearity leaks.
For the random construction, by contrast, the memory itself carries the state away from its input subspace.

In the nonlinear closed-loop RC (Eq.~\eqref{eq:closed_loop_main}), we now examine to what extent the attractor reconstructed near this subspace is actually attracting. We quantify this by basin stability \cite{Menck2013-ex} (\(S_B\)): the fraction of randomly perturbed reservoir states whose closed-loop predictions converge back to the true attractor (Methods~\ref{methods-basin}).
One realization of the random construction loses nearly a third of its basin (\(S_B=0.71\); Fig.~\figlink{fig:bs}{g}), while both proposed constructions retain it entirely (\(S_B=1.00\); Figs.~\figlink{fig:bs}{h,i}).
Whether the random construction retains its basin depends on the realization (Extended Data Fig.~\ref{fig:ed-sweep-basin}), whereas every realization of the proposed constructions keeps the basin nearly full (\(S_B\ge0.97\)).

Through this input-layer design, the reconstructed attractor is not indiscriminately embedded in a high-dimensional state space but placed near a subspace that remains transversally attracting under the closed-loop dynamics.
This design principle is what we call \emph{attractor reconstruction in attracting subspaces} (ARAS).

\begin{figure}[H]
    \centering
    \includegraphics[width=0.99\linewidth,height=\textheight,keepaspectratio]{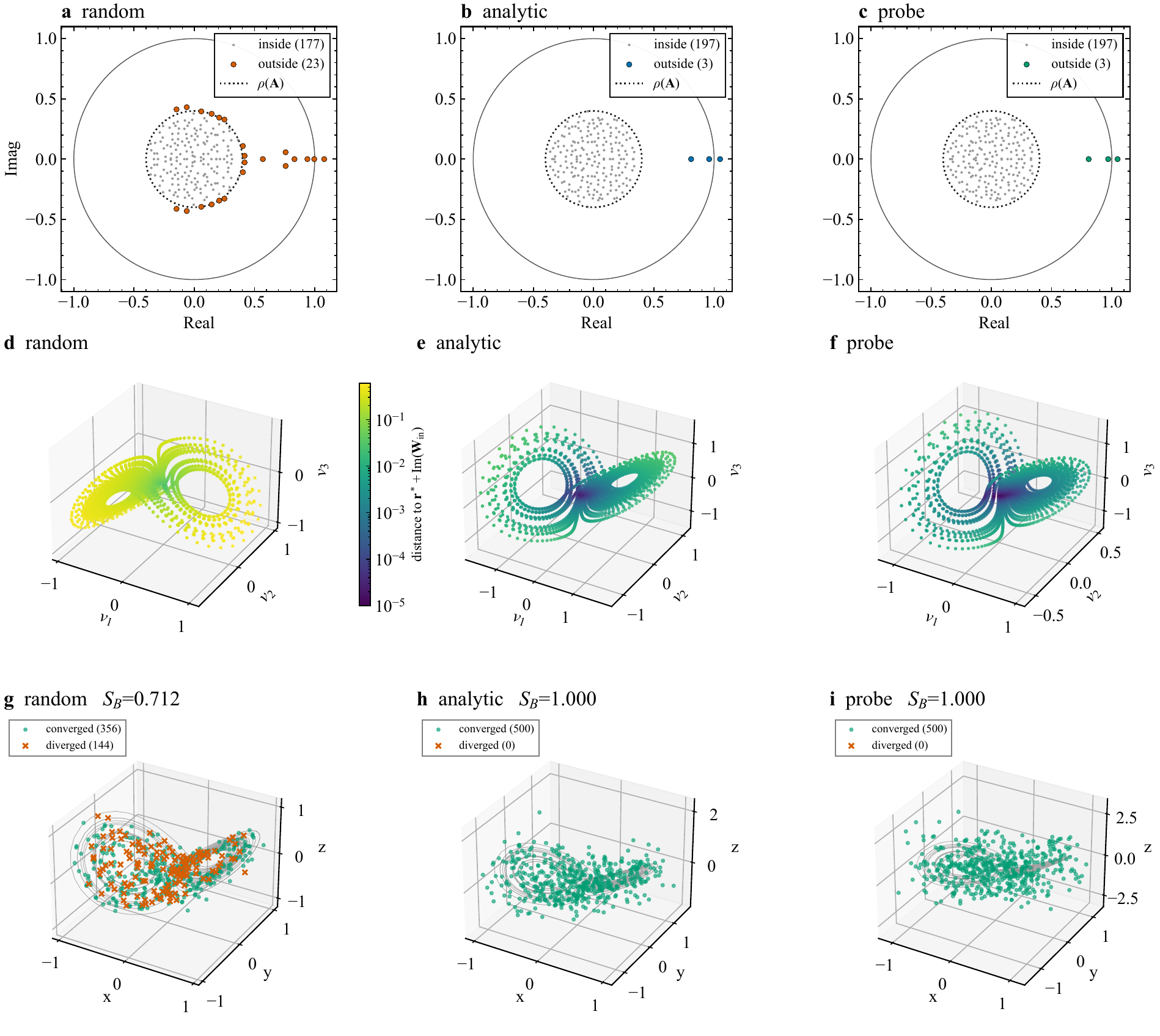}
    \caption{
    \textbf{Reconstruction geometry and nonlinear stability for the three input-layer constructions} (random, analytic, and probe-based) on the Lorenz system.
    \textbf{a}--\textbf{c}, Eigenvalue spectra of the closed-loop matrix \(\mathbf{W}_{\mathrm{cl}}\) in the complex plane; the green dashed circle marks the spectral radius \(\rho\) of \(\mathbf{A}\). The random construction (\textbf{a}) produces spurious slow modes that escape beyond \(\rho\), whereas the analytic (\textbf{b}) and probe-based (\textbf{c}) constructions move exactly \(D'=3\) eigenvalues outward while anchoring the remaining \(N-D'\) inside \(\rho\), exactly for the analytic construction (Theorem~\ref{thm:invariant_eigenvalues}) and approximately for the probe-based one.
    \textbf{d}--\textbf{f}, Reconstructed attractor in the reservoir state space, projected onto the affine subspace \(\mathbf{r}^{*}+\mathrm{Im}(\mathbf{W}_{\mathrm{in}})\) (the designed input subspace anchored at the input-free equilibrium \(\mathbf{r}^{*}\)) and colored by the distance of each state to that plane (logarithmic scale). The random reconstruction (\textbf{d}) lies far from its input subspace (mean \(\approx 0.25\)), while both proposed constructions (\textbf{e},\textbf{f}) place the attractor within a thin neighborhood of it (mean \(\approx 0.003\)--\(0.006\)).
    \textbf{g}--\textbf{i}, Basin-stability analysis (nonlinear dynamic stability): 500 reservoir states, perturbed by uniform noise, are mapped to the observation space (\(\hat{\mathbf{u}}_{t_0}=\mathbf{W}_{\mathrm{out}}\mathbf{r}_{t_0}\)). initial points whose closed-loop trajectories returned to the true attractor are marked with green dots, those that did not with orange crosses, and \(S_B\) is the converged fraction (random \(0.71\); analytic and probe-based \(1.00\)).}
    \label{fig:bs}
\end{figure}

\subsection{Transverse contraction dominates along the reconstruction}\label{normal-hyperbolicity}

We now turn to a further property of this reconstruction: it is \emph{normally attracting} \cite{Fenichel1979-ej, Hirsch1977-ry}.
This property makes the stability of the reconstruction robust not only to transverse perturbations of the internal state, as quantified by basin stability, but also to perturbations of the model itself, such as training error and small changes in the internal weights.
Roughly speaking, an invariant set of a dynamical system is normally attracting if the dynamics contracts sufficiently faster in the directions transverse to the set than in the tangential directions (Methods~\ref{methods-lyapunov}).
When this holds, the attracting structure is guaranteed to persist under sufficiently small \(C^{1}\) perturbations of the map \cite{Fenichel1979-ej, Hirsch1977-ry}; in particular, the reconstruction remains stable under small changes in the internal parameters.

To verify this property, we measured both rates directly along the driven reservoir trajectory over short stretches of \(k\) steps: the transverse strongest growth rate (weakest contraction) \(\Lambda_{k}^{\perp}\) and the strongest tangential contraction \(\Lambda_{k}^{T,\min}\) (Methods~\ref{methods-lyapunov}).
The reconstruction is normally attracting when two conditions hold at every point: their difference, the domination metric \(\Delta_{k}=\Lambda_{k}^{\perp}-\Lambda_{k}^{T,\min}\), is negative, and the transverse contraction itself is strong (\(\Lambda_{k}^{\perp}\ll0\)).
For both proposed constructions \(\Delta_{16}\) is negative at \emph{every} sampled point of the attractor (window \(k=16\), about a fifth of a Lyapunov time), with a margin of \(0.32\)--\(0.33\) per step (Figs.~\figlink{fig:nh}{a--c}).
The transverse contraction is also strong: \(\Lambda_{k}^{\perp}\) is negative at every sampled point from window \(k=8\) on and deepens with \(k\), to \(\Lambda_{16}^{\perp}\le-0.62\)---transverse deviations shrink by nearly a factor of two per step (Figs.~\figlink{fig:nh}{d--f}, see also Extended Data Fig.~\ref{fig:ed-sweep-dom}).
For the random construction neither condition holds anywhere: \(\Delta_{16}\) is negative nowhere, and \(\Lambda_{k}^{\perp}\) remains positive at every point at all windows (Figs.~\figlink{fig:nh}{a,d}).
The measurement therefore confirms, point by point, that the proposed constructions achieve the normally attracting structure intended by the anchoring.

Normal attraction predicts that, under small changes in the internal weights, the reconstruction persists as a nearby invariant set.
We tested this prediction by perturbing the closed-loop matrix to \(\mathbf{W}_{\mathrm{cl}}+\mathbf{W}_{\varepsilon}\), with \(\mathbf{W}_{\varepsilon}\) a random matrix of spectral norm \(\varepsilon\), and running the perturbed closed loop autonomously (Methods~\ref{methods-persistence}).
At \(\varepsilon=1\times10^{-4}\), the analytic and probe-based reconstructions retain the attractor in \(93\%\) and \(87\%\) of runs, whereas the random reconstruction loses it in \(74\%\) (Extended Data Fig.~\ref{fig:ed-eps-bars}).
Figures~\figlink{fig:nh}{g--i} show a representative example, the most frequent joint outcome.

\begin{figure}[H]
    \centering
    \includegraphics[width=0.99\linewidth,height=\textheight,keepaspectratio]{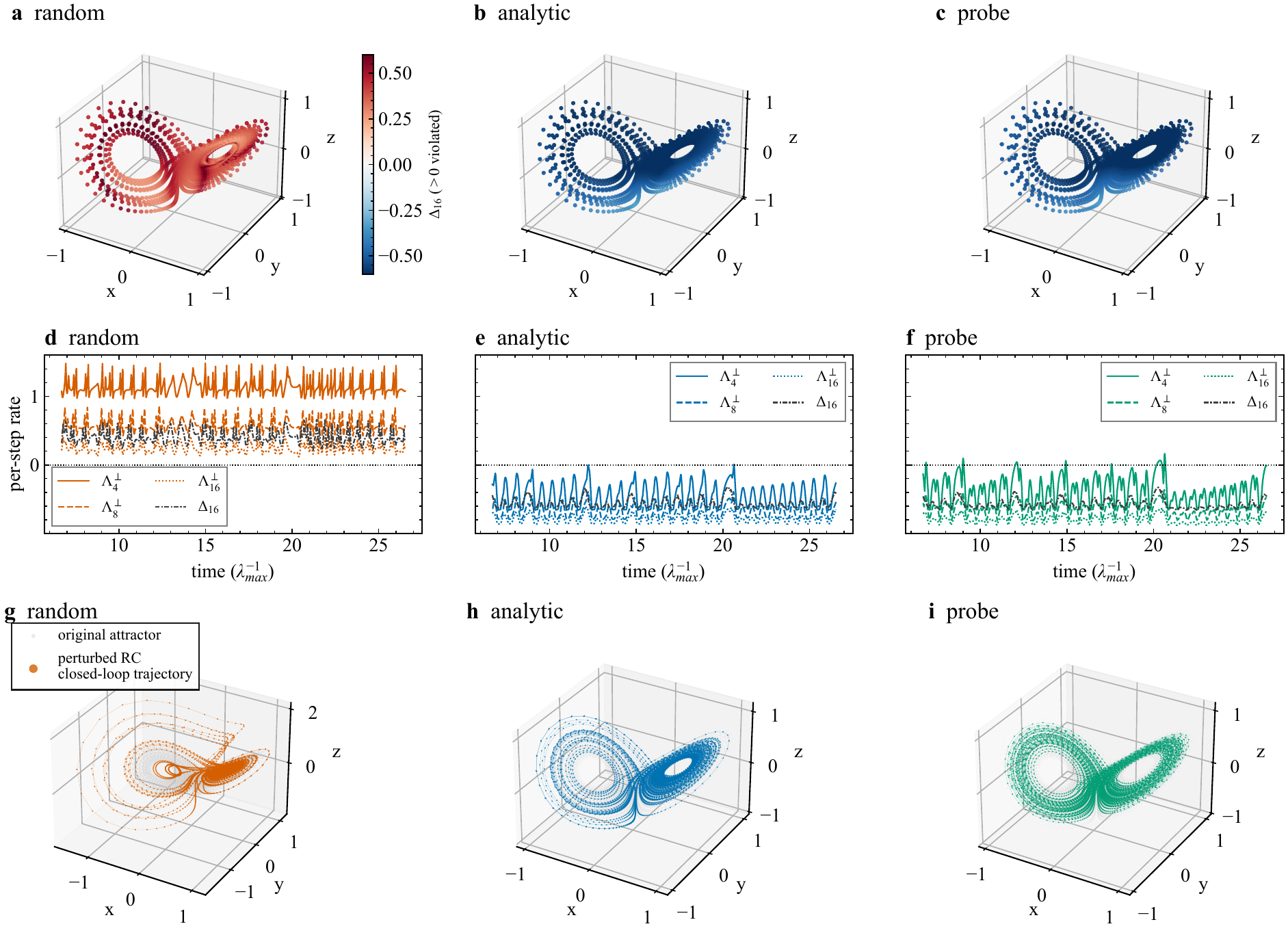}
    \caption{\textbf{Transverse contraction dominates along the reconstruction} for the three constructions (Lorenz, \(N=200\); \(\rho=0.4\)).
    \textbf{a}--\textbf{c}, The reconstructed attractor in the observation space, each point colored by the \(k=16\) domination metric \(\Delta_{16}=\Lambda_{16}^{\perp}-\Lambda_{16}^{T,\min}\) (color scale centered at zero). Negative (blue) means the transverse contraction dominates the tangential dynamics at each point; positive (red) means domination is violated. The random construction (\textbf{a}) is non-dominated everywhere, the analytic (\textbf{b}) and probe-based (\textbf{c}) constructions dominated everywhere.
    \textbf{d}--\textbf{f}, Windowed transverse growth rates \(\Lambda_k^{\perp}\) for \(k=4,8,16\) and the \(k=16\) domination metric \(\Delta_{16}\) (dashed) along the \(2{,}000\)-step driven trajectory (first \(500\) steps discarded as transient; Methods~\ref{methods-lyapunov}), for the random (\textbf{d}), analytic (\textbf{e}), and probe-based (\textbf{f}) constructions. For the proposed constructions the rates are negative at every point from \(k=8\) on and deepen with increasing window; for the random construction they remain positive at every point at all windows.
    \textbf{g}--\textbf{i}, Persistence under a perturbation of the internal weights. The closed loop is run autonomously after perturbing the closed-loop matrix to \(\mathbf{W}_{\mathrm{cl}}+\mathbf{W}_{\varepsilon}\), where \(\mathbf{W}_{\varepsilon}\) is a random matrix of spectral norm \(\varepsilon=1\times10^{-4}\). Shown is one reservoir realization with one draw of \(\mathbf{W}_{\varepsilon}\), plotted for \(5{,}000\) steps (\(\approx 67\,\lambda_{1}^{-1}\)). The random reconstruction (\textbf{g}) makes repeated excursions far outside the attractor, whereas the analytic (\textbf{h}) and probe-based (\textbf{i}) reconstructions remain within a small neighborhood of it. Over 20 reservoir realizations and 10 draws of \(\mathbf{W}_{\varepsilon}\), the analytic and probe-based reconstructions retain the attractor in \(93\%\) and \(87\%\) of runs at this \(\varepsilon\), and the random one retains it in only \(26\%\) (see also Extended Data Figs.~\ref{fig:ed-eps-bars} and \ref{fig:ed-eps-ladder}).}
    \label{fig:nh}
\end{figure}

\subsection{Predicting and reproducing chaos}\label{anchoring-improves-reproduction-of-chaotic-dynamics}

Anchoring the fast modes deliberately restricts the number of learnable time scales in the closed-loop dynamics. A natural concern is that this constraint costs the ability to reproduce the target dynamics.
We therefore compared closed-loop prediction across the three constructions on the Lorenz system \cite{Lorenz1963-kf} (\(D=3\)) over 20 reservoir realizations, using the valid prediction time (VPT)~\cite{Vlachas2020-ob}---the duration for which the prediction error stays below a fixed threshold, in units of the Lyapunov time \(\lambda_{1}^{-1}\) (Methods~\ref{methods-metrics}). The random construction attains a median VPT of \(12.4\,\lambda_{1}^{-1}\), but erratically: individual realizations range from \(2.4\) to \(17.0\). Both proposed constructions are more accurate and far more consistent---analytic \(17.2\) (range \(16.7\)--\(17.8\)), probe-based \(17.6\) (range \(16.9\)--\(18.0\))---with every realization at \(16.7\,\lambda_{1}^{-1}\) or above (Figs.~\figlink{fig:vptls}{a--c}). These values are competitive with those reported for hyperparameter-tuned RCs \cite{Pathak2017-id, Platt2021-el, Bollt2021-mj, Lu2018-vk, Hurley2025-ue}, yet the proposed constructions require no extensive tuning, maintaining high performance across a broad range of hyperparameters such as spectral radii of \(\mathbf{A}\), reservoir sizes, and reservoir topologies (Extended Data Fig.~\ref{fig:ed-sweep-vpt} and Supplementary Secs.~\ref{si:grid} and \ref{si-radius-sweep}).

Short-term accuracy, however, does not imply fidelity to the target system.
To assess the reproduction of the system's long-term behavior, we turn to the Lyapunov spectrum---the time-averaged rates at which the dynamics expands or contracts nearby trajectories along the attractor.
The Lorenz system has one positive, one zero, and one negative exponent (\(\lambda_{1},\lambda_2,\lambda_3=0.91,0,-14.57\)). A closed loop \eqref{eq:closed_loop_main} that reproduces these values is faithful to the target in this sense.
We therefore estimated the leading exponents of each trained closed loop (Figs.~\figlink{fig:vptls}{d--f}).
Here the random construction fails on almost every realization, including those that predict well.
Even when the positive and zero exponents are captured, spurious exponents are interposed, and the negative exponent is almost never recovered.
With such spurious exponents present, closed-loop trajectories can diverge from the reconstruction and behave entirely differently from the target system (Fig.~\figlink{fig:fig1}{d}).
Both proposed constructions instead recover the entire Lyapunov spectrum, including the negative exponent, on every realization, and no spurious exponents are interposed. This reflects the mechanism verified in the previous section: the anchored \(N-D'\) fast eigenvalues make trajectories decay toward the learned slow subspace, so the reconstruction remains near the strongly attracting slow manifold and inherits the true relaxation rates---the negative Lyapunov exponents---toward the attractor. The data-driven recovery of negative Lyapunov exponents has long been considered difficult \cite{Sauer1998-la}. Our design overcomes this difficulty, enabling reliable estimation of the fractal dimension of the attractor (the Kaplan--Yorke dimension) \cite{Kaplan1979-jy}.

\begin{figure}[H]
    \centering
    \includegraphics[width=0.99\linewidth,height=\textheight,keepaspectratio]{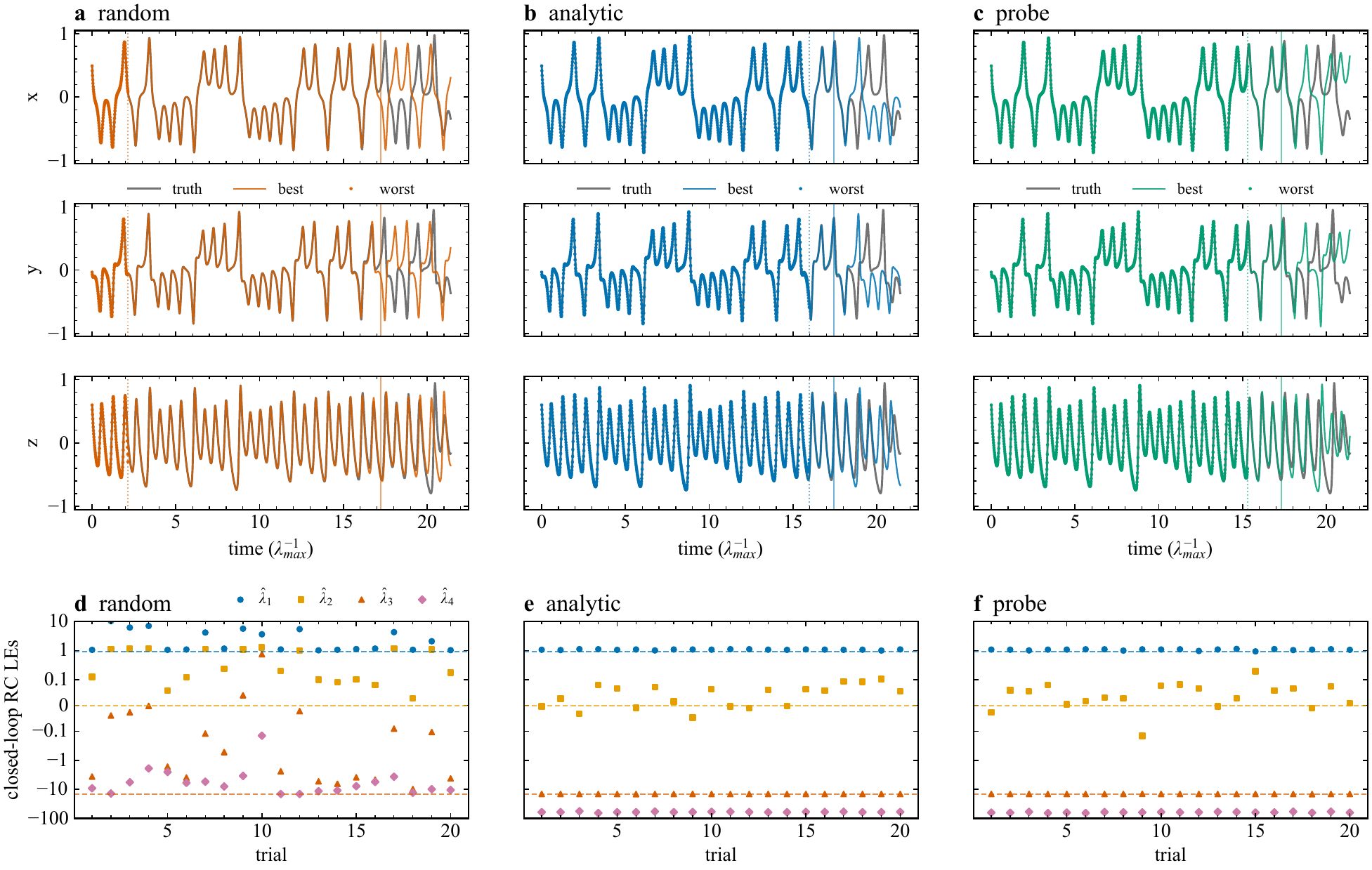}
    \caption{\textbf{Closed-loop prediction and Lyapunov spectrum estimation} (Lorenz system).
        \textbf{a}--\textbf{c}, Autonomous closed-loop predictions of the three observed components \((x,y,z)\) against the ground truth (gray) for the random (\textbf{a}), analytic (\textbf{b}), and probe-based (\textbf{c}) constructions. For each construction the best realization (dashed, full length) and the worst realization (dotted, drawn up to its VPT) are shown; vertical lines mark the VPT of the worst (dotted) and best (dashed) realizations.
        \textbf{d}--\textbf{f}, Estimated Lyapunov exponents \(\hat\lambda_1\)--\(\hat\lambda_4\) for each reservoir realization for the random (\textbf{d}), analytic (\textbf{e}), and probe-based (\textbf{f}) constructions (symmetric log scale, linear within \(\pm0.1\)); dashed lines are the true values \(\lambda_{1,2,3}=0.91,0,-14.57\). The random construction inflates the leading exponent and almost never recovers the negative one. Both proposed constructions recover the qualitative structure of the spectrum, including the negative exponent, on every realization.}
    \label{fig:vptls}
\end{figure}

To test whether an RC designed by the proposed method remains effective beyond the Lorenz system, we evaluated its performance on \texttt{dysts}, a benchmark dataset of diverse low-dimensional chaotic dynamical systems \cite{Gilpin2021-wp, Gilpin2023-ua} (Methods~\ref{benchmark-evaluation-details}).
We evaluated prediction with the VPT and with the benchmark's primary metric, the symmetric mean absolute percentage error (SMAPE, Methods~\ref{methods-metrics}).
As the baseline, we took the best of the 24 models in Gilpin's benchmark \cite{Gilpin2023-ua} (NBEATS~\cite{Oreshkin2020-av}/NHiTS~\cite{Challu2022-tr}, \(13.9 \pm 2.8\,\lambda_{1}^{-1}\)) and matched its experimental conditions (Extended Data Table~\ref{table:para}).
Under these conditions, the probe-based construction achieved an average SMAPE prediction horizon of \(14.1 \pm 14.0\,\lambda_{1}^{-1}\) and the analytic construction \(13.9 \pm 13.9\,\lambda_{1}^{-1}\) (Extended Data Fig.~\ref{fig:ed-benchmark-smape}), matching the best reference models and roughly doubling the random construction (\(6.8 \pm 10.1\,\lambda_{1}^{-1}\); Extended Data Table~\ref{table:dysts}).
The proposed constructions are thus competitive with the reference models, but the variance across systems is far larger (\(\pm 14.0\) against \(\pm 2.8\)). They approximate some systems well and fail on others.
With the probe-based construction, 44 of the 134 systems sustained accurate prediction beyond \(20\,\lambda_{1}^{-1}\), while 46 fell below \(1\,\lambda_{1}^{-1}\) (Extended Data Figs.~\ref{fig:ed-benchmark-smape} and \ref{fig:si-heatmap}).
Our RC used one and the same configuration for all 134 systems, with only \(D'\) following the observation dimension (Methods~\ref{benchmark-evaluation-details}).
The spread is the price of a single fixed configuration.

\section{Discussion}\label{discussion}

This study shows that a fast-slow separation can be introduced into a data-driven dynamical-system learner, the reservoir computer, in advance and in a task-agnostic manner.
We confirmed that restricting the number of slow modes available to learning, while keeping all remaining modes fast, suppresses the emergence of spurious slow modes.
We found that reconstructing the attractor on the slow subspace created by the fast-slow gap---attractor reconstruction in attracting subspaces (ARAS)---leads to robust reproduction.
Moreover, installing this gap with an ample margin makes the reconstruction normally attracting \cite{Fenichel1979-ej, Hirsch1977-ry, Kuehn2015-pl}, so the reproduction of the dynamics is robust to training error and to small changes in the internal weights.
The reproduction is also robust to the choice of the hyperparameters and of the invariant subspace (Supplementary Figs.~\ref{fig:si-eigsel}--\ref{fig:si-grid-le}).

RC is increasingly implemented as ``physical RC'' on diverse substrates \cite{Tanaka2019-vn}, where the reservoir is typically a black box.
The probe-based construction requires neither \(\mathbf{A}\) itself nor knowledge of the nonlinearity of the physical reservoir.
It assumes only that the input layer can be engineered and that the input--output response of the reservoir can be measured (Methods~\ref{probe-based-construction}).
In our experiments, it attained the stability and accuracy of the analytic construction. The resulting physical RC is thus expected to be robust to input and internal noise, and to gradual drift of the internal weights of the device.

The failure mode of conventional RC identified in this work can be restated in control-theoretic terms.
With a random input layer, the pair \((\mathbf{A}, \mathbf{W}_{\mathrm{in}})\) is generically controllable \cite{Kailath1980-ou, Grigoryeva2023-cw}, which means that the feedback \(\mathbf{W}_{\mathrm{in}}\mathbf{W}_{\mathrm{out}}\) installed by training can, in principle, relocate all \(N\) eigenvalues of \(\mathbf{A}\).
This freedom is exactly what allows spurious slow modes to arise, and our design removes it deliberately.
The pair becomes uncontrollable, only \(D'\) eigenvalues remain learnable, and the remaining \(N-D'\) can never turn slow.
Controllability is ordinarily a desirable property. The embedding theory of RC employs it as a sufficient condition for embedding the attractor \cite{Hart2020-kx, Grigoryeva2023-cw}, and network science seeks it by identifying driver nodes \cite{Liu2011-qs}.
For the closed-loop reproduction of dynamics, however, appropriately restricting controllability is beneficial.

Finally, the principle extends naturally to other frameworks that model dynamics in a representation space, such as state-space models~\cite{Gu2021-uh, Gu2023-xy}, sparse identification of nonlinear dynamics~\cite{Brunton2016-gc}, and (extended) dynamic mode decomposition~\cite{Schmid2010-lj, Williams2015-jg}.
To reproduce real-world dynamics robustly, the attractor should be reconstructed on a submanifold of the minimal dimension needed to represent the structure of the target faithfully, and this submanifold should be kept attracting (ARAS).
We propose designing learning machines on this basis as a key step toward reliable prediction, data-driven modeling, and physical AI.

\section{Methods}\label{methods}

\subsection{Reservoir computing fundamentals and closed-loop operation}
\label{reservoir-computing-fundamentals-and-closed-loop-operation}

We use a standard discrete-time RC model \cite{Jaeger2001-vs, Jaeger2007-py, Maass2002-me, Lukosevicius2009-vt}.
Consider a target dynamical system \(f\) evolving on an invariant set \(I\), where the state \(\mathbf{x}_t \in I\) is observed as a \(D\)-dimensional vector \(\mathbf{u}_t = \mathbf{h}(\mathbf{x}_t)\) through the observation map \(\mathbf{h}: I \to \mathbb{R}^D \).
The RC receives this input time series \(\mathbf{u}_t\) and updates its \(N\)-dimensional reservoir state \(\mathbf{r}_t \in \mathbb{R}^N\) according to:
\begin{equation}\label{eq:rc_update}
    \mathbf{r}_{t+1} = \bar{f}(\mathbf{r}_t, \mathbf{W}_{\mathrm{in}}\mathbf{u}_t) = \boldsymbol{\tau}(\mathbf{A} \mathbf{r}_t + \mathbf{W}_{\mathrm{in}} \mathbf{u}_t + \mathbf{b}),
\end{equation}
where \(\mathbf{A} \in \mathbb{R}^{N \times N}\), \(\mathbf{W}_{\mathrm{in}} \in \mathbb{R}^{N \times D}\), and \(\mathbf{b} \in \mathbb{R}^N\) denote the reservoir transition matrix, the input weight matrix, and the bias vector, respectively.
We adopted the hyperbolic tangent \(\boldsymbol{\tau}(\cdot) = \boldsymbol{\tanh}(\cdot)\), the standard choice of activation function in echo state networks \cite{Jaeger2001-vs}, a widely used form of RC implemented as recurrent neural networks (RNNs).
These internal parameters are fixed rather than trained, and in the standard configuration they are generated randomly.
The hyperparameter values used in this study are listed in Extended Data Table~\ref{table:para}.

The theoretical foundation for ESN's success in modeling dynamical systems lies in the echo state property (ESP) \cite{Jaeger2001-vs, Jaeger2007-py}.
ESP signifies that, when driven by a common input signal, the reservoir states asymptotically converge to a unique trajectory, regardless of the initial reservoir state.
The driven reservoir state then becomes a function of the target state, \(\mathbf{r}_t = \Phi(\mathbf{x}_t)\) as \(t \to \infty\) \cite{Lu2018-vk, Hart2020-kx, Grigoryeva2021-mg}.
If \(\Phi\) is an embedding (a diffeomorphism onto its image), a set topologically equivalent to the target invariant set is reconstructed within the reservoir state space.

The output weight matrix \(\mathbf{W}_{\mathrm{out}}\) is trained on the driven reservoir states to predict the next observation \(\mathbf{u}_{t+1}\).
Using ridge regression with a regularization parameter \(\beta\), the optimal weights are given by:
\begin{equation}\label{eq:readout_learning}
    \mathbf{W}_{\mathrm{out}} = \mathbf{Y} \mathbf{R}^T (\mathbf{R} \mathbf{R}^T + \beta \mathbf{I})^{-1},
\end{equation}
where \(\mathbf{R}\) and \(\mathbf{Y}\) denote the matrices collecting the reservoir states \(\mathbf{r}_t\) and the targets \(\mathbf{u}_{t+1}\), respectively, over \(T_{\mathrm{train}}\) time steps following a warmup period of \(T_{\mathrm{warm}}\) steps that washes out the dependence on the initial reservoir state.

After training, the RC performs autonomous time-series generation (closed-loop prediction).
This is achieved by feeding back the model's output \(\hat{\mathbf{u}}_t = \mathbf{W}_{\mathrm{out}}\mathbf{r}_t\) as the input for the next time step.
The dynamics of the closed-loop system are then described as the following autonomous dynamical system with respect to the reservoir state \(\mathbf{r}_t\):

\begin{equation}\label{eq:closed_loop}
    \mathbf{r}_{t+1} = \boldsymbol{\tau}(\mathbf{A} \mathbf{r}_t + \mathbf{W}_{\mathrm{in}}\mathbf{W}_{\mathrm{out}}\mathbf{r}_t + \mathbf{b} ) = \boldsymbol{\tau}( \mathbf{W}_{\mathrm{cl}}\mathbf{r}_t + \mathbf{b} ) =: \hat{f}(\mathbf{r}_t),
\end{equation}
where we define \(\mathbf{W}_{\mathrm{cl}} := \mathbf{A} + \mathbf{W}_{\mathrm{in}}\mathbf{W}_{\mathrm{out}}\) as the closed-loop matrix.

\subsection{Configuration of Reservoir configuration}
\label{methods-reservoir-config}

The RC model configuration follows the formulation of Eqs.~\eqref{eq:rc_update} and \eqref{eq:closed_loop}.
Unless otherwise specified, the common hyperparameters listed in Extended Data Table~\ref{table:para} were used in the experiments.
The reservoir transition matrix \(\mathbf{A}\) was generated as an Erd\H{o}s--R\'enyi random graph with connection probability \(p_A\), its nonzero entries drawn from a uniform distribution, and rescaled so that its spectral radius matched the target value \(\rho\).
For the input weight matrix \(\mathbf{W}_{\mathrm{in}}\), we compared the conventional random initialization, whose elements were drawn from a uniform distribution, with our proposed constructions (Secs.~\ref{deterministic-input-layer-construction-proposed-method} and \ref{probe-based-construction}).
In every construction, each column of \(\mathbf{W}_{\mathrm{in}}\) was normalized to unit norm.
Because the input data were normalized to \([-1, 1]\), we did not tune the input scaling.
The output weight matrix \(\mathbf{W}_{\mathrm{out}}\) was optimized via ridge regression as shown in Eq.~\eqref{eq:readout_learning}.
The bias vector was constant, \(\mathbf{b} = b_0\mathbf{1}\) with \(b_0 = 0.01\). The initial state \(\mathbf{r}_0\) was sampled from a uniform distribution.
The dependence on these hyperparameters is examined in the Supplementary hyperparameter grid (Supplementary Figs.~\ref{fig:si-grid} and \ref{fig:si-grid-le}).

\subsection{Analytic construction for a known reservoir (proposed method)}\label{deterministic-input-layer-construction-proposed-method}

To address the instability that stems from spurious slow modes of the closed-loop matrix \(\mathbf{W}_{\mathrm{cl}}\), we design the input layer so as to restrict the number of eigenvalues of \(\mathbf{A}\) that the training perturbation \(\mathbf{W}_{\mathrm{in}}\mathbf{W}_{\mathrm{out}}\) can move.
Specifically, we construct the columns of the input weight matrix \(\mathbf{W}_{\mathrm{in}} \in \mathbb{R}^{N \times D}\) from an invariant subspace of the reservoir transition matrix \(\mathbf{A} \in \mathbb{R}^{N \times N}\).
From the eigenvectors of \(\mathbf{A}\), we select \(D'\) linearly independent real basis vectors \(\{\mathbf{v}_1, \dots, \mathbf{v}_{D'}\}\) spanning a \(D'\)-dimensional real \(\mathbf{A}\)-invariant subspace \(V_{D'} = \mathrm{span}\{\mathbf{v}_1, \dots, \mathbf{v}_{D'}\}\) (\(D' \ge D\)).
Here \(D\) is the input dimension, and \(D'\) corresponds to the number of eigenvalues that training can move, as explained below.
To construct this invariant subspace, the real eigenvectors are selected directly for the real eigenvalues.
For complex conjugate eigenvalue pairs, the real and imaginary parts of the eigenvectors are used to keep the basis real-valued.\footnote{The invariant subspace can also be obtained directly from an ordered real Schur decomposition of \(\mathbf{A}\), without computing eigenvectors. This also covers the case where \(\mathbf{A}\) is not diagonalizable, in which generalized eigenvectors replace the eigenvectors in the construction above.}
Any \(\mathbf{W}_{\mathrm{in}}\) whose columns lie in \(V_{D'}\) leaves the remaining \(N-D'\) eigenvalues of the closed-loop matrix unchanged, as guaranteed by Theorem~\ref{thm:invariant_eigenvalues} below.
Throughout this work we set \(D' = D\) (\(=3\) in the Lorenz experiments; the two coincide only when all state variables are observed).
Collecting the basis vectors into \(\mathbf{V}_{D'} = [\mathbf{v}_1, \dots, \mathbf{v}_{D'}] \in \mathbb{R}^{N \times D'}\), we construct the input weight matrix as
\begin{equation}\label{eq:win_construction}
    \mathbf{W}_{\mathrm{in}} = \mathbf{V}_{D'}\mathbf{K},
\end{equation}
where \(\mathbf{K} \in \mathbb{R}^{D' \times D}\) is an arbitrary coefficient matrix of full column rank, so that the column space of \(\mathbf{W}_{\mathrm{in}}\) lies within the invariant subspace: \(\mathrm{Im}(\mathbf{W}_{\mathrm{in}}) \subseteq V_{D'}\).
In our experiments we took \(\mathbf{K} = \operatorname{diag}(\|\mathbf{v}_1\|^{-1}, \dots, \|\mathbf{v}_D\|^{-1})\), so that the columns of \(\mathbf{W}_{\mathrm{in}}\) are the basis vectors normalized to unit norm.

Given this construction of \(\mathbf{W}_{\mathrm{in}}\), we establish the following theorem regarding the spectral properties of the closed-loop matrix \(\mathbf{W}_{\mathrm{cl}} = \mathbf{A} + \mathbf{W}_{\mathrm{in}}\mathbf{W}_{\mathrm{out}}\).

\begin{thm}[Invariant Eigenvalues]\label{thm:invariant_eigenvalues}
    Let \(\mathbf{A} \in \mathbb{R}^{N \times N}\) be a real matrix and let \(V_{D'} \subseteq \mathbb{R}^{N}\) be a \(D'\)-dimensional real \(\mathbf{A}\)-invariant subspace. Let \(\{\mu_1, \dots, \mu_{D'}\}\) be the eigenvalues of \(\mathbf{A}\) corresponding to \(V_{D'}\), and \(\{\mu_{D'+1}, \dots, \mu_N\}\) the remaining eigenvalues. If a matrix \(\mathbf{B} \in \mathbb{R}^{N \times D}\) satisfies \(\mathrm{Im}(\mathbf{B}) \subseteq V_{D'}\) (i.e., the column vectors of \(\mathbf{B}\) are contained in \(V_{D'}\)), then for any real matrix \(\mathbf{C} \in \mathbb{R}^{D \times N}\), the spectrum of \(\mathbf{A}' = \mathbf{A} + \mathbf{B}\mathbf{C}\) contains the multiset \(\{\mu_{D'+1}, \dots, \mu_N\}\).\footnote{The theorem assumes only the \(\mathbf{A}\)-invariance of \(V_{D'}\). Diagonalizability of \(\mathbf{A}\) is not required.}
\end{thm}
\begin{proof}
    The proof is provided in Supplementary Information.
\end{proof}

This theorem guarantees that when \(\mathbf{W}_{\mathrm{in}}\) is constructed from a \(D'\)-dimensional real invariant subspace \(V_{D'}\) of \(\mathbf{A}\) as in the proposed method, \(N-D'\) eigenvalues of the closed-loop matrix \(\mathbf{W}_{\mathrm{cl}}\) remain equal to those of \(\mathbf{A}\), regardless of the learned \(\mathbf{W}_{\mathrm{out}}\). Therefore, designing \(\mathbf{A}\) in advance with a small spectral radius \(\rho\) keeps all anchored eigenvalues fast and contracting.

\begin{cor}[Inheritance of Invariant Subspace]
    \label{cor:inheritance}
    Under the assumptions of Theorem \ref{thm:invariant_eigenvalues}, the subspace \(V_{D'}\) remains invariant under \(\mathbf{A}' = \mathbf{A} + \mathbf{B}\mathbf{C}\).
\end{cor}

Corollary~\ref{cor:inheritance} states that the invariant subspace \(V_{D'}\) used to construct the input layer coincides with the invariant subspace of \(\mathbf{W}_{\mathrm{cl}}\) that corresponds to the eigenvalues moved by training.
It follows from a direct computation.
Since \(V_{D'}\) is \(\mathbf{A}\)-invariant, there exists \(\mathbf{X} \in \mathbb{R}^{D' \times D'}\) such that \(\mathbf{A}\mathbf{V}_{D'} = \mathbf{V}_{D'}\mathbf{X}\), and hence
\begin{equation}\label{eq:aras}
    \mathbf{W}_{\mathrm{cl}}\mathbf{V}_{D'} = (\mathbf{A} + \mathbf{W}_{\mathrm{in}}\mathbf{W}_{\mathrm{out}})\mathbf{V}_{D'} = \mathbf{V}_{D'}(\mathbf{X} + \mathbf{K}\mathbf{W}_{\mathrm{out}}\mathbf{V}_{D'}),
\end{equation}
so \(V_{D'}\) is invariant under \(\mathbf{W}_{\mathrm{cl}}\).
The \(D'\) eigenvalues that training can move are those of \(\mathbf{W}_{\mathrm{cl}}\) corresponding to \(V_{D'}\), represented by \(\mathbf{X} + \mathbf{K}\mathbf{W}_{\mathrm{out}}\mathbf{V}_{D'}\), whereas the remaining \(N-D'\) eigenvalues stay at those of \(\mathbf{A}\) (Theorem \ref{thm:invariant_eigenvalues}) and are contracting by design.
In this linear picture, \(V_{D'}\) is an exponentially attracting slow subspace of the reservoir state space: trajectories decay toward \(V_{D'}\) at the fast rates set by \(\mathbf{A}\) and evolve within it under the \(D'\) learned slow modes.

\subsection{Probe-based construction for an unknown reservoir}\label{probe-based-construction}

Even when the reservoir transition matrix \(\mathbf{A}\) is inaccessible, as is typical for physical reservoirs, we can still design an input layer that induces the same fast-slow structure as in the previous section, using only the reservoir's input--output response.
We first drive the input-free reservoir to an equilibrium \(\mathbf{r}^{*} = \bar{f}(\mathbf{r}^{*}, \mathbf{0})\) and linearize the open-loop map \eqref{eq:rc_update} there, obtaining the state Jacobian \(\mathbf{J}^{*} = \partial\mathbf{r}_{t+1}/\partial\mathbf{r}_t\) and the input Jacobian \(\mathbf{B}^{*} = \partial\mathbf{r}_{t+1}/\partial\mathbf{u}_t\), both evaluated at \((\mathbf{r}^{*}, \mathbf{0})\).
We estimate \(\mathbf{J}^{*}\) by probing the black-box reservoir and design \(\mathbf{W}_{\mathrm{in}}\) so that the image of \(\mathbf{B}^{*}\) lies within a \(D'\)-dimensional \(\mathbf{J}^{*}\)-invariant subspace (see Supplementary Information~\ref{si:probe} for the full procedure).
By Theorem~\ref{thm:invariant_eigenvalues}, this anchors the \(N-D'\) transverse eigenvalues of the closed-loop linearization at \(\mathbf{r}^{*}\) to the contracting eigenvalues of \(\mathbf{J}^{*}\), and thereby induces a fast-slow structure in closed-loop operation.
The construction relies only on measured quantities and requires no knowledge of the explicit form of the open-loop map \eqref{eq:rc_update}.

\subsection{Basin stability}
\label{methods-basin}

To evaluate basin stability \cite{Menck2013-ex}, we randomly selected states from the \(2{,}000\) teacher-driven reservoir states \(\{\mathbf{r}_t\}\) and perturbed each by uniform noise drawn from \([-10^{-3}, 10^{-3}]^N\) to define the initial points.  After a transient of \(2{,}000\) steps of closed-loop operation, we tested whether the subsequent \(1{,}000\) steps of the closed-loop output \(\{\hat{\mathbf{u}}_t\}\) remain close to the true attractor \(\{\mathbf{u}_t\}\) in the observation space, by checking whether the directed Hausdorff distance from the prediction to the attractor satisfies \(\max_{t} \min_{t'} \|\hat{\mathbf{u}}_t - \mathbf{u}_{t'}\| < 0.25\). We generated 500 perturbed initial points and defined \(S_B\) as the fraction of them whose closed-loop trajectories return to the true attractor.

\subsection{Prediction metrics}
\label{methods-metrics}

Short-term prediction performance was evaluated by performing closed-loop prediction and calculating the valid prediction time (VPT) \cite{Vlachas2020-ob}.
VPT is defined as the time duration, normalized by the maximal Lyapunov exponent \(\lambda_{1}\), until the normalized root-mean-square error (NRMSE) exceeds a threshold of 0.5:
\[
    \text{NRMSE}(\hat{\mathbf{u}}_t) = \sqrt{ \frac{1}{D} \sum_{i=1}^{D} \frac{(\hat{u}_{t,i} - u_{t,i})^2}{\sigma_i^2} },
\]
\[
    \text{VPT} = \frac{1}{\lambda_{1}} \operatorname{argmax}_{t_f} \{ t_f \mid \text{NRMSE}(\hat{\mathbf{u}}_t) < 0.5, \forall t \le t_f \},
\]
where \(D\) is the dimension of the observation and \(\sigma_i^2\) denotes the variance of the \(i\)-th dimension of the target time series.

For the large-scale benchmark involving 134 systems, we also employed the cumulative SMAPE metric \cite{Gilpin2023-ua} defined as:
\[
    \epsilon_t \equiv \frac{200}{t} \sum_{t'=1}^{t} \frac{\|\mathbf{u}_{t'} - \hat{\mathbf{u}}_{t'}\|}{\|\mathbf{u}_{t'}\| + \|\hat{\mathbf{u}}_{t'}\|}.
\]
The SMAPE prediction horizon is the time, in units of \(\lambda_{1}^{-1}\), at which \(\epsilon_t\) first exceeds \(50\) .
Both metrics were evaluated over \(N_{\mathrm{test}}\) distinct prediction start points (Extended Data Table~\ref{table:para}), and the average over the start points was reported.

\subsection{Lyapunov spectrum and dominated-splitting test}
\label{methods-lyapunov}

The Lyapunov spectrum of the closed-loop RC was estimated using the standard QR-decomposition procedure \cite{Shimada1979-mv, Benettin1980-tk}, applied to the Jacobian of the closed-loop map \eqref{eq:closed_loop}, \(\mathrm{D}\hat{f}(\mathbf{r}_t) = \operatorname{diag}\!\big(1 - \boldsymbol{\tanh}^2(\mathbf{W}_{\mathrm{cl}}\mathbf{r}_t + \mathbf{b})\big)\,\mathbf{W}_{\mathrm{cl}} = \operatorname{diag}(1 - \mathbf{r}_{t+1}^2)\, \mathbf{W}_{\mathrm{cl}}\), along the teacher-driven trajectory \(\{\mathbf{r}_t\}\) generated by Eq.~\eqref{eq:open_loop}.
From this spectrum we evaluate whether the closed loop faithfully reproduces the chaotic character of the target system, by checking that the \(N\) closed-loop exponents contain the Lyapunov exponents of the target attractor.
The exponents that remain after removing the target ones are the transverse Lyapunov exponents (TLEs) \cite{Fujisaka1983-hu, Lu2018-vk}.
They quantify the transverse stability of the reconstructed attractor, and a positive TLE signals transverse instability.
Moreover, when every TLE lies below the most negative target exponent, as the proposed constructions achieve (Fig.~\figlink{fig:vptls}{e,f}), the transverse stability is especially strong. The map carrying the target attractor onto the reconstructed attractor is then differentiable---the defining property of strong generalized synchronization between the drive (target) and the response (reservoir) \cite{Rulkov1995-tq, Kocarev1996-mt, Hunt1997-aw}.

In the sweep experiments, the accuracy of the reproduced spectrum is summarized by the Lyapunov spectrum error (LSE),
\begin{equation}\label{eq:lse}
    \mathrm{LSE} = \sum_{i\,:\,\lambda_i \ge 0} \bigl|\hat\lambda_i-\lambda_i\bigr| + \sum_{i\,:\,\lambda_i < 0} \frac{\bigl|\hat\lambda_i-\lambda_i\bigr|}{|\lambda_i|},
\end{equation}
where \(\lambda_i\) and \(\hat\lambda_i\) are the Lyapunov exponents of the target system and of the closed-loop RC, respectively, both in descending order.
The non-negative exponents enter through their absolute errors and the negative ones through their relative errors.
Capturing the non-negative part of the spectrum matters more for the dynamics than matching the strongly contracting negative exponents exactly, and the relative normalization keeps their large magnitudes from dominating the metric.
For the Lorenz system, this reads
\[
    \mathrm{LSE} = |\hat\lambda_1-\lambda_1| + |\hat\lambda_2-\lambda_2| + \frac{|\hat\lambda_3-\lambda_3|}{|\lambda_3|}.
\]

We next describe how the dominated splitting reported in Sec.~\ref{normal-hyperbolicity} was tested. Using the same QR procedure, we propagated an orthonormal \(D'\)-frame \(\mathbf{Q}_t \in \mathbb{R}^{N\times D'}\) with the Jacobian cocycle \(\mathrm{D}\hat{f}(\mathbf{r}_t)\) along the \(2{,}000\)-step teacher-driven trajectory, discarding the first \(500\) steps as the transient of an arbitrary initial frame; \(\mathbf{Q}_t\) then spans the \(D'\) leading backward-Lyapunov directions \cite{Kuptsov2012-jj}. In this frame, the tangential block is \(\mathbf{T}_t=\mathbf{Q}_{t+1}^{\top}\mathrm{D}\hat{f}(\mathbf{r}_t)\,\mathbf{Q}_t\) and the transverse block is \(\mathbf{T}_t^{\perp}=(\mathbf{I}-\mathbf{Q}_{t+1}\mathbf{Q}_{t+1}^{\top})\,\mathrm{D}\hat{f}(\mathbf{r}_t)\,(\mathbf{I}-\mathbf{Q}_t\mathbf{Q}_t^{\top})\). Over a window of \(k\) steps we defined the maximal transverse growth rate \(\Lambda_k^{\perp}(t)=\tfrac1k\log\|\mathbf{T}_{t+k-1}^{\perp}\cdots\mathbf{T}_t^{\perp}\|_2\), the minimal tangential growth rate (strongest contraction rate) \(\Lambda_k^{T,\min}(t)=\tfrac1k\log\sigma_{\min}(\mathbf{T}_{t+k-1}\cdots\mathbf{T}_t)\), and the domination metric \(\Delta_k=\Lambda_k^{\perp}-\Lambda_k^{T,\min}\). We report \(\Lambda_k^{\perp}\) at \(k=4,8,16\) and \(\Delta_k\) at \(k=16\).

Combined with the strong transverse contraction (\(\sup_t \Lambda_k^{\perp}\ll 0\)), uniform domination at window \(k\) (\(\sup_t \Delta_k(t)<0\)) makes the reconstruction a partially hyperbolic invariant set \cite{Hirsch1977-ry}. Theorem~5.5 of \cite{Hirsch1977-ry} provides, through every point of such a compact set, a locally invariant manifold tangent to the slow bundle spanned by \(\mathbf{Q}_t\). The transverse contraction makes these locally invariant manifolds attracting.

We state the precise scope of this test. What we have established is partial hyperbolicity on a finite sample of the reconstructed attractor, which certifies a locally invariant attracting manifold through each sampled point. This does not guarantee the existence of a global attracting manifold containing the reconstructed attractor. Establishing that global object would require constructing a candidate manifold containing the reconstructed attractor and verifying normal attraction at every point of that manifold. Following the standard route of geometric singular perturbation theory \cite{Fenichel1979-ej, Kuehn2015-pl}, the closed-loop map \eqref{eq:closed_loop} can be decomposed into fast and slow coordinates---the time-scale separation set by the anchored spectral gap---and adiabatic elimination of the fast coordinates yields a critical manifold, a natural such candidate. The concrete computation lies beyond the scope of this paper, but we conjecture that this critical manifold, and the normally attracting slow submanifold it induces, lie near \(V_{D'}\) while deviating from it, possibly as graphs over \(V_{D'}\). Local or global, reconstructing the attractor on a normally attracting slow submanifold---the ARAS strategy---is effective for data-driven modeling because it can exploit the benefits of normal attraction, notably persistence under small perturbations.

\subsection{Persistence test}\label{methods-persistence}

We perturb the trained closed loop of the common Lorenz configuration (Extended Data Table~\ref{table:para}) by adding to \(\mathbf{W}_{\mathrm{cl}}\) a random matrix \(\mathbf{W}_{\varepsilon}\) with independent standard normal entries, rescaled so that \(\lVert\mathbf{W}_{\varepsilon}\rVert_{2}=\varepsilon\).
Starting from the teacher-forced state at the end of training, the perturbed closed loop \(\mathbf{r}_{t+1} = \boldsymbol{\tanh}\big((\mathbf{W}_{\mathrm{cl}}+\mathbf{W}_{\varepsilon})\mathbf{r}_t + \mathbf{b}\big)\) runs autonomously for \(5{,}000\) steps (\(\approx 67\,\lambda_{1}^{-1}\)).
Each run is classified by whether it stays on the attractor and continues to visit both wings, with the first \(2{,}000\) steps treated as a transient and excluded.
We repeat this over 20 reservoir realizations and 10 independent draws of \(\mathbf{W}_{\varepsilon}\), shared across the realizations and the levels \(\varepsilon \in \{1,\dots,10\}\times10^{-4}\), together with an unperturbed control for every realization.
The outcome frequencies and a single-configuration \(\varepsilon\) ladder are shown in Extended Data Figs.~\ref{fig:ed-eps-bars} and \ref{fig:ed-eps-ladder}.

\subsection{Benchmark data}
\label{methods-data}
\label{benchmark-evaluation-details}

All time series in this study were generated with the \texttt{dysts} benchmark library \cite{Gilpin2021-wp, Gilpin2023-ua}, which provides trajectories of 135 diverse low-dimensional chaotic dynamical systems together with their integration settings and largest Lyapunov exponents.
The detailed analyses presented in the Results section use the Lorenz-63 system \cite{Lorenz1963-kf} from this collection,
\begin{equation}\label{eq:lorenz}
    \dot{x} = 10\,(y - x), \qquad
    \dot{y} = x\,(28 - z) - y, \qquad
    \dot{z} = xy - \tfrac{8}{3}\,z,
\end{equation}
at the standard parameter values, with all three variables observed (\(\mathbf{u}_t = (x, y, z)\), \(D=3\)).
Prior to feeding the data into the RC, each state variable was preprocessed via dimension-wise linear scaling to map the values to the range \([-1, 1]\).

For the benchmark of Sec.~\ref{anchoring-improves-reproduction-of-chaotic-dynamics}, we followed the protocol of Gilpin \cite{Gilpin2021-wp, Gilpin2023-ua} and used 134 of the 135 catalogued systems, comprising 113 continuous autonomous, 6 delayed autonomous, and 15 non-autonomous equations.
Metrics were averaged over 10 realizations of the reservoir configuration.
The model is non-leaky (the leaking rate is fixed to unity) and uses fixed parameters throughout, with a symmetric reservoir matrix \(\mathbf{A}\) (Extended Data Table~\ref{table:para}).
In contrast, in Gilpin's benchmark the lookback window---the hyperparameter that sets the time scale of the model---was tuned for every system--model pair \cite{Gilpin2023-ua}.
The invariant-subspace dimension followed the observation dimension, \(D' = D\), and \(D' = D + 1\) for the non-autonomous systems, because an additional slow mode is required to represent the driving term.
For the training data size, we adopted the settings of Gilpin's echo state network model, with warmup steps \(T_{\mathrm{warm}} = 100\) and training steps \(T_{\mathrm{train}} = 900\).
Further details are given in Supplementary Sec.~\ref{si-benchmark-details}.

\backmatter
\bmhead{Acknowledgements}
This work was supported by JSPS KAKENHI Grant Number JP25KJ1714, JST ALCA-Next Grant No. JPMJAN23F2 and JST Moonshot R\&D Grant No. JPMJMS2021, JST PRESTO Grant No. JPMJPR25K5, JST CREST Grant No. JPMJCR25R1.

\section{Supplementary information} \label{supplementary-information}

\subsection{Theoretical foundation}
\label{theoretical-foundation}

\setcounter{thm}{0}

For completeness, we restate Theorem~\ref{thm:invariant_eigenvalues} from Methods~\ref{deterministic-input-layer-construction-proposed-method} and provide its proof.

\begin{thm}[Invariant Eigenvalues]\label{thm:invariant_eigenvalues_si}
    Let \(\mathbf{A} \in \mathbb{R}^{N \times N}\) be a real matrix and let \(V_{D'} \subseteq \mathbb{R}^{N}\) be a \(D'\)-dimensional real \(\mathbf{A}\)-invariant subspace. Let \(\{\mu_1, \dots, \mu_{D'}\}\) be the eigenvalues of \(\mathbf{A}\) corresponding to \(V_{D'}\), and \(\{\mu_{D'+1}, \dots, \mu_N\}\) the remaining eigenvalues. If a matrix \(\mathbf{B} \in \mathbb{R}^{N \times D}\) satisfies \(\mathrm{Im}(\mathbf{B}) \subseteq V_{D'}\) (i.e., the column vectors of \(\mathbf{B}\) are contained in \(V_{D'}\)), then for any real matrix \(\mathbf{C} \in \mathbb{R}^{D \times N}\), the spectrum of \(\mathbf{A}' = \mathbf{A} + \mathbf{B}\mathbf{C}\) contains the multiset \(\{\mu_{D'+1}, \dots, \mu_N\}\).
\end{thm}

\begin{proof}
    Choose \(\mathbf{P}_{D'} \in \mathbb{R}^{N \times D'}\) whose columns form a basis of \(V_{D'}\), and extend it to an invertible matrix \(\mathbf{P} = [\mathbf{P}_{D'}, \mathbf{P}_{N-D'}] \in \mathbb{R}^{N \times N}\).
    Since \(V_{D'}\) is \(\mathbf{A}\)-invariant, the similarity transformation \(\tilde{\mathbf{A}} = \mathbf{P}^{-1} \mathbf{A} \mathbf{P}\) assumes a block upper triangular form:
    \[
        \tilde{\mathbf{A}} =
        \begin{pmatrix}
            \mathbf{A}_{11} & \mathbf{A}_{12} \\
            \mathbf{0}      & \mathbf{A}_{22}
        \end{pmatrix},
    \]
    where the spectrum of \(\mathbf{A}_{22}\) is the complementary multiset \(\{\mu_{D'+1}, \dots, \mu_N\}\).

    Since \(\mathrm{Im}(\mathbf{B}) \subseteq V_{D'} = \mathrm{Im}(\mathbf{P}_{D'})\), the transformed matrix \(\tilde{\mathbf{B}} = \mathbf{P}^{-1} \mathbf{B}\) takes the form:
    \[
        \tilde{\mathbf{B}} = \begin{pmatrix} \mathbf{B}_1 \\ \mathbf{0} \end{pmatrix},
    \]
    where \(\mathbf{B}_1 \in \mathbb{R}^{D' \times D}\).
    Partition \(\mathbf{C} \mathbf{P}\) as \([\mathbf{C}_1, \mathbf{C}_2]\), where \(\mathbf{C}_1 \in \mathbb{R}^{D \times D'}\) and \(\mathbf{C}_2 \in \mathbb{R}^{D \times (N-D')}\).
    The similarity transformation of \(\mathbf{A}' = \mathbf{A} + \mathbf{B}\mathbf{C}\) is then given by:
    \[
        \begin{aligned}
            \tilde{\mathbf{A}}' & = \mathbf{P}^{-1} (\mathbf{A} + \mathbf{B}\mathbf{C}) \mathbf{P}                                                                                                                                                       \\
                                & = \tilde{\mathbf{A}} + \tilde{\mathbf{B}} (\mathbf{C} \mathbf{P})                                                                                                                                                      \\
                                & = \begin{pmatrix} \mathbf{A}_{11} & \mathbf{A}_{12} \\ \mathbf{0} & \mathbf{A}_{22} \end{pmatrix} + \begin{pmatrix} \mathbf{B}_1 \\ \mathbf{0} \end{pmatrix} \begin{pmatrix} \mathbf{C}_1 & \mathbf{C}_2 \end{pmatrix} \\
                                & = \begin{pmatrix} \mathbf{A}_{11} & \mathbf{A}_{12} \\ \mathbf{0} & \mathbf{A}_{22} \end{pmatrix} + \begin{pmatrix} \mathbf{B}_1 \mathbf{C}_1 & \mathbf{B}_1 \mathbf{C}_2 \\ \mathbf{0} & \mathbf{0} \end{pmatrix}     \\
                                & = \begin{pmatrix} \mathbf{A}_{11} + \mathbf{B}_1 \mathbf{C}_1 & \mathbf{A}_{12} + \mathbf{B}_1 \mathbf{C}_2 \\ \mathbf{0} & \mathbf{A}_{22} \end{pmatrix}.
        \end{aligned}
    \]
    The spectrum of the block triangular matrix \(\tilde{\mathbf{A}}'\) is the union, with multiplicities, of the spectra of its diagonal blocks.
    The bottom-right block remains exactly \(\mathbf{A}_{22}\) regardless of \(\mathbf{C}\).
    Thus, the spectrum of \(\mathbf{A}'\) always contains the multiset \(\{\mu_{D'+1}, \dots, \mu_N\}\).
\end{proof}

\subsection{Details of the probe-based construction}
\label{si:probe}

This section gives the full procedure summarized in Methods~\ref{probe-based-construction}. The design is driven by the measured Jacobian \(\mathbf{J}^{*}\) in place of the reservoir transition matrix \(\mathbf{A}\), which need not be accessible for a black-box reservoir.

\noindent
\textbf{Linearization at the input-free equilibrium.}\;
We first drive the reservoir with zero input until it settles on the input-free equilibrium \(\mathbf{r}^{*} = \bar{f}(\mathbf{r}^{*}, \mathbf{0})\), found by iterating the map to convergence. We write the reservoir update \eqref{eq:rc_update} as \(\mathbf{r}_{t+1}=\bar{f}(\mathbf{r}_t,\mathbf{v}_t)\), a general black-box map taking the injected input \(\mathbf{v}_t=\mathbf{W}_{\mathrm{in}}\mathbf{u}_t\), with \(\mathbf{r}^{*}=\bar{f}(\mathbf{r}^{*},\mathbf{0})\). Linearization about \((\mathbf{r}^{*},\mathbf{u}_t\equiv\mathbf{0})\) gives
\begin{equation}\label{eq:probe_linearization}
    \delta\mathbf{r}_{t+1}=\mathbf{J}^{*}\,\delta\mathbf{r}_t+\mathbf{B}^{*}\,\delta\mathbf{u}_t,\qquad
    \mathbf{J}^{*}=\frac{\partial \bar{f}}{\partial\mathbf{r}}\bigg|_{*},\quad \mathbf{G}:=\frac{\partial \bar{f}}{\partial\mathbf{v}}\bigg|_{*},\quad \mathbf{B}^{*}=\mathbf{G}\mathbf{W}_{\mathrm{in}},
\end{equation}
all Jacobians evaluated at \((\mathbf{r}^{*},\mathbf{0})\). Only the state and input Jacobians \(\mathbf{J}^{*}\) and \(\mathbf{B}^{*}\) are accessible through the input--output response, and the construction below uses only these.

\noindent
\textbf{Markov-parameter estimation.}\;
The open-loop response of the linearization \eqref{eq:probe_linearization} about \(\mathbf{r}^{*}\) is a convolution of the past inputs,
\[
    \delta\mathbf{r}_{t+1}=\sum_{k\ge0}\mathbf{M}_k\,\mathbf{u}_{t-k},\qquad \mathbf{M}_k=\mathbf{J}^{*k}\mathbf{B}^{*},
\]
with Markov parameters \(\mathbf{M}_k\) (in particular \(\mathbf{M}_0=\mathbf{B}^{*}\), \(\mathbf{M}_1=\mathbf{J}^{*}\mathbf{B}^{*}\)). We probe the reservoir with balanced pseudo-random binary sequences of amplitude \(\varepsilon\) (each of the \(2^{D}\) sign patterns equally often), running a \emph{paired} rollout from \(\mathbf{r}^{*}\), one driven by \(+\mathbf{u}_t\) and one by \(-\mathbf{u}_t\), and keeping the half-difference \(\tfrac12(\mathbf{r}^{+}_t-\mathbf{r}^{-}_t)\), which cancels the even-order nonlinearity and isolates \(\delta\mathbf{r}_t\). Truncating the sum at lag \(L\) and least-squares regressing \(\delta\mathbf{r}_{t+1}\) on \((\mathbf{u}_t,\dots,\mathbf{u}_{t-L})\) yields \(\mathbf{M}_0,\dots,\mathbf{M}_L\). We assume the reservoir is designed to be fast, with \(\mathbf{J}^{*}\) of small spectral radius (or norm) and hence quickly decaying Markov parameters, so that a finite lag \(L\) absorbs the memory and \(\mathbf{M}_0,\mathbf{M}_1\) are unbiased.

A single \(D\)-column probe excites only \(D\) of the \(N\) state directions, so we repeat the measurement with \(P=\lceil N/D\rceil\) probe input matrices \(\mathbf{W}_{\mathrm{in}}^{\mathrm{ref},1},\dots,\mathbf{W}_{\mathrm{in}}^{\mathrm{ref},P}\), their columns drawn from a Sylvester--Hadamard matrix so that together they span \(\mathbb{R}^{N}\), and stack the per-probe estimates into \(\mathbf{M}_0^{\mathrm{stack}},\mathbf{M}_1^{\mathrm{stack}}\in\mathbb{R}^{N\times PD}\). The shift relation \(\mathbf{M}_1=\mathbf{J}^{*}\mathbf{M}_0\) then reads \(\mathbf{M}_1^{\mathrm{stack}}=\mathbf{J}^{*}\mathbf{M}_0^{\mathrm{stack}}\), which we solve in the least-squares sense for the estimate \(\widehat{\mathbf{J}}^{*}\).

\noindent
\textbf{Input-layer construction.}\;
The design target is the closed-loop Jacobian evaluated at the input-free equilibrium \(\mathbf{r}^{*}\),\footnote{The closed-loop dynamics has in general its own fixed point \(\mathbf{r}_{\mathrm{cl}}^{*}=\hat{f}(\mathbf{r}_{\mathrm{cl}}^{*})\neq\mathbf{r}^{*}\). We evaluate the closed-loop linearization at the input-free equilibrium \(\mathbf{r}^{*}\), where \(\mathbf{J}^{*}\) and \(\mathbf{B}^{*}\) are measured. The two coincide when \(\mathbf{W}_{\mathrm{out}}\mathbf{r}^{*}=\mathbf{0}\), and are close for the near-origin equilibria used here.} which is
\begin{equation}\label{eq:probe_jcl}
    \mathbf{J}_{\mathrm{cl}}=\mathbf{J}^{*}+\mathbf{B}^{*}\mathbf{W}_{\mathrm{out}},\qquad \mathbf{B}^{*}=\mathbf{G}\mathbf{W}_{\mathrm{in}},
\end{equation}
so Theorem~\ref{thm:invariant_eigenvalues}, applied to \(\mathbf{J}^{*}+\mathbf{B}^{*}\mathbf{W}_{\mathrm{out}}\), shows that if the input Jacobian \(\mathbf{B}^{*}\) spans a \(D'\)-dimensional \(\mathbf{J}^{*}\)-invariant subspace \(V_J\), the remaining \(N-D'\) eigenvalues of \(\mathbf{J}_{\mathrm{cl}}\) stay at those of \(\mathbf{J}^{*}\), which are fast and contracting by design. The columns of \(\mathbf{V}_{D'}\) are a real basis of \(V_J\), the \(D'\)-dimensional invariant subspace of \(\widehat{\mathbf{J}}^{*}\) selected as in the analytic construction (with \(D'=D\) in this work).

To make \(\mathbf{B}^{*}=\mathbf{G}\mathbf{W}_{\mathrm{in}}\) span \(V_J\) we set \(\mathbf{W}_{\mathrm{in}}=\mathbf{G}^{-1}\mathbf{V}_{D'}\), giving \(\mathbf{B}^{*}=\mathbf{V}_{D'}\) exactly. Here \(\mathbf{G}^{-1}\) maps the desired input Jacobian to the input weights that produce it. The gain \(\mathbf{G}=\partial \bar{f}/\partial\mathbf{v}|_{*}\) (the derivative of the reservoir map with respect to the injected input at \(\mathbf{r}^{*}\)) is itself measured: each reference probe gives \(\mathbf{M}_0^{(p)}=\mathbf{G}\mathbf{W}_{\mathrm{in}}^{\mathrm{ref},p}\), so stacking the \(P\) probes and solving \(\mathbf{M}_0^{\mathrm{stack}}=\mathbf{G}\,[\mathbf{W}_{\mathrm{in}}^{\mathrm{ref},1}|\cdots|\mathbf{W}_{\mathrm{in}}^{\mathrm{ref},P}]\) for \(\mathbf{G}\).\footnote{For the elementwise activation used here \(\mathbf{G}=\mathrm{diag}(\mathbf{1}-\mathbf{r}^{*2})\) is diagonal, so a single reference probe suffices and each \(G_{ii}\) follows from a row-wise ratio \(G_{ii}=\sum_j(\mathbf{M}_0)_{ij}(\mathbf{W}_{\mathrm{in}}^{\mathrm{ref}})_{ij}\big/\sum_j(\mathbf{W}_{\mathrm{in}}^{\mathrm{ref}})_{ij}^{2}\).}
Every quantity entering the design is thus measured from the input--output response, and the explicit form of the open-loop map \eqref{eq:rc_update} is never needed.

\begin{table}[h!]
    \centering
    \caption{\textbf{Probe-based construction hyperparameters} used for the Lorenz reservoirs (\(N=200\), \(D=3\)) of Figs.~\ref{fig:bs}--\ref{fig:vptls}.}\label{table:probe}
    \begin{tabular}{@{}lll@{}}
        \toprule
        Symbol & Meaning & Value \\
        \midrule
        \(\varepsilon\) & Probe amplitude & \(3\times10^{-6}\) \\
        \(L\) & Regression lag (input taps) & 20 \\
        -- & PRBS sequences per probe & 3 \\
        -- & Steps per sequence & 1200 \\
        \(P\) & Number of Hadamard probes & \(\lceil N/D\rceil=67\) \\
        \(D'\) & Invariant-subspace dimension (\(=D\) here) & 3 (1 conjugate pair) \\
        \bottomrule
    \end{tabular}
\end{table}

\subsection{Eigenvalue-selection dependence}\label{si:eigsel}

\setcounter{figure}{0}
\renewcommand{\thefigure}{S\arabic{figure}}

We tested whether the performance of the analytic construction depends on which invariant subspace of \(\mathbf{A}\) is selected.
The representative Lorenz setting of Figs.~\ref{fig:bs}--\ref{fig:vptls} was used throughout (\(N=200\), \(D'=D=3\), \(\rho=0.4\)).
A real three-dimensional invariant subspace of a real matrix can be spanned in only two ways: by three real eigenvectors, or by one complex-conjugate pair (its real and imaginary parts) together with one real eigenvector.
Both types were swept: all 56 selections of the first type, and 101 selections of the second type, stratified by the mean modulus of the selected eigenvalues.
For each of the 157 selections, VPT was evaluated over 20 reservoir realizations, and the Lyapunov exponents \(\hat\lambda_1\)--\(\hat\lambda_4\) over one realization.

The dependence is flat (Supplementary Fig.~\ref{fig:si-eigsel}).
Across all 157 selections, the per-selection mean VPT lands in a band only \(0.9\,\lambda_{1}^{-1}\) wide (\(17.0\)--\(17.9\,\lambda_{1}^{-1}\); median standard deviation \(0.4\)).
The two selection types give nearly the same mean VPT (\(17.6\pm0.5\) and \(17.6\pm0.4\)).
All three target exponents are reproduced accurately for every selection, and \(\hat\lambda_4\) never rises above \(-48\).
No spurious slow mode appears for any admissible choice of the invariant subspace.

\begin{figure}[H]
    \centering
    \includegraphics[width=0.99\linewidth,height=\textheight,keepaspectratio]{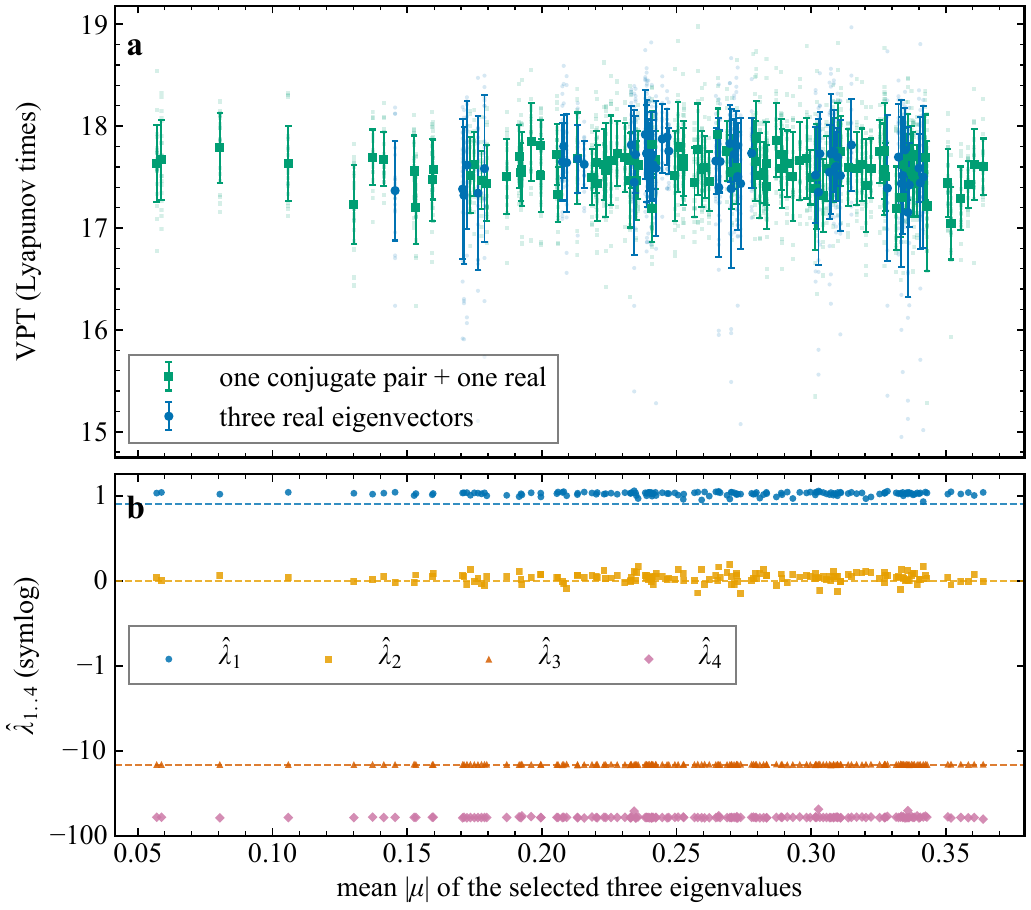}
    \caption{\textbf{Reconstruction quality does not depend on which invariant subspace is selected.}
    Each point is one three-dimensional invariant subspace of \(\mathbf{A}\), used to build \(\mathbf{W}_{\mathrm{in}}\) by the analytic construction at the representative setting (Lorenz, \(N=200\), \(\rho=0.4\)); the horizontal axis is the mean modulus of the selected three eigenvalues.
    \textbf{a}, VPT over 20 reservoir realizations (markers, mean; bars, standard deviation; faint dots, individual realizations) for the two admissible subspace types: three real eigenvectors (blue dots) and one conjugate pair plus one real eigenvector (green squares).
    \textbf{b}, Estimated Lyapunov exponents of the trained closed loop for a single realization, on a symmetric-log axis; dashed lines mark the true values \(\lambda_{1,2,3}=0.91, 0, -14.57\). All LEs are recovered for all 157 subspaces and \(\hat\lambda_4\) stays below \(-48\): no spurious slow mode appears for any admissible choice.}
    \label{fig:si-eigsel}
\end{figure}

\subsection{Dependence on the invariant-subspace dimension}\label{si:dprime}

We varied the invariant-subspace dimension \(D'\) from the minimal \(D'=3\) to the unconstrained limit \(D'=N=200\), at the same representative setting.
The invariant subspace was built from the eigenvalues with the largest real parts, keeping complex-conjugate pairs together.
Each column of \(\mathbf{W}_{\mathrm{in}}\) was drawn as a random linear combination of orthonormal basis vectors of this subspace.
The statement of Theorem~\ref{thm:invariant_eigenvalues} was verified directly at every \(D'\): the \(N-D'\) unselected eigenvalues of \(\mathbf{A}\) reappear in the spectrum of \(\mathbf{W}_{\mathrm{cl}}\), and exactly \(D'\) eigenvalues are moved (Supplementary Figs.~\figlink{fig:si-dprime}{a--f}).

Increasing \(D'\) degrades the performance through spurious slow modes (Supplementary Figs.~\figlink{fig:si-dprime}{g--m}).
In the prediction horizon, the degradation appears only from \(D'=50\) on.
Every realization keeps a VPT above \(11\,\lambda_{1}^{-1}\) at \(D'=3\) and \(10\), whereas realizations with VPT below \(10\,\lambda_{1}^{-1}\) appear at larger \(D'\) (4 of 20 at \(D'=50\), 9 of 20 at \(D'=200\); Supplementary Fig.~\figlink{fig:si-dprime}{g}).
The Lyapunov spectrum is more sensitive.
At \(D'=3\), every realization recovers all target exponents.
As \(D'\) grows, spurious Lyapunov exponents gradually appear (Supplementary Figs.~\figlink{fig:si-dprime}{h--m}).
These exponents come from slow modes that are unnecessary for reproducing the Lorenz dynamics but are admitted by the larger subspace.
Keeping \(D'\) minimal is thus part of the design.

\begin{figure}[H]
    \centering
    \includegraphics[width=0.99\linewidth,height=\textheight,keepaspectratio]{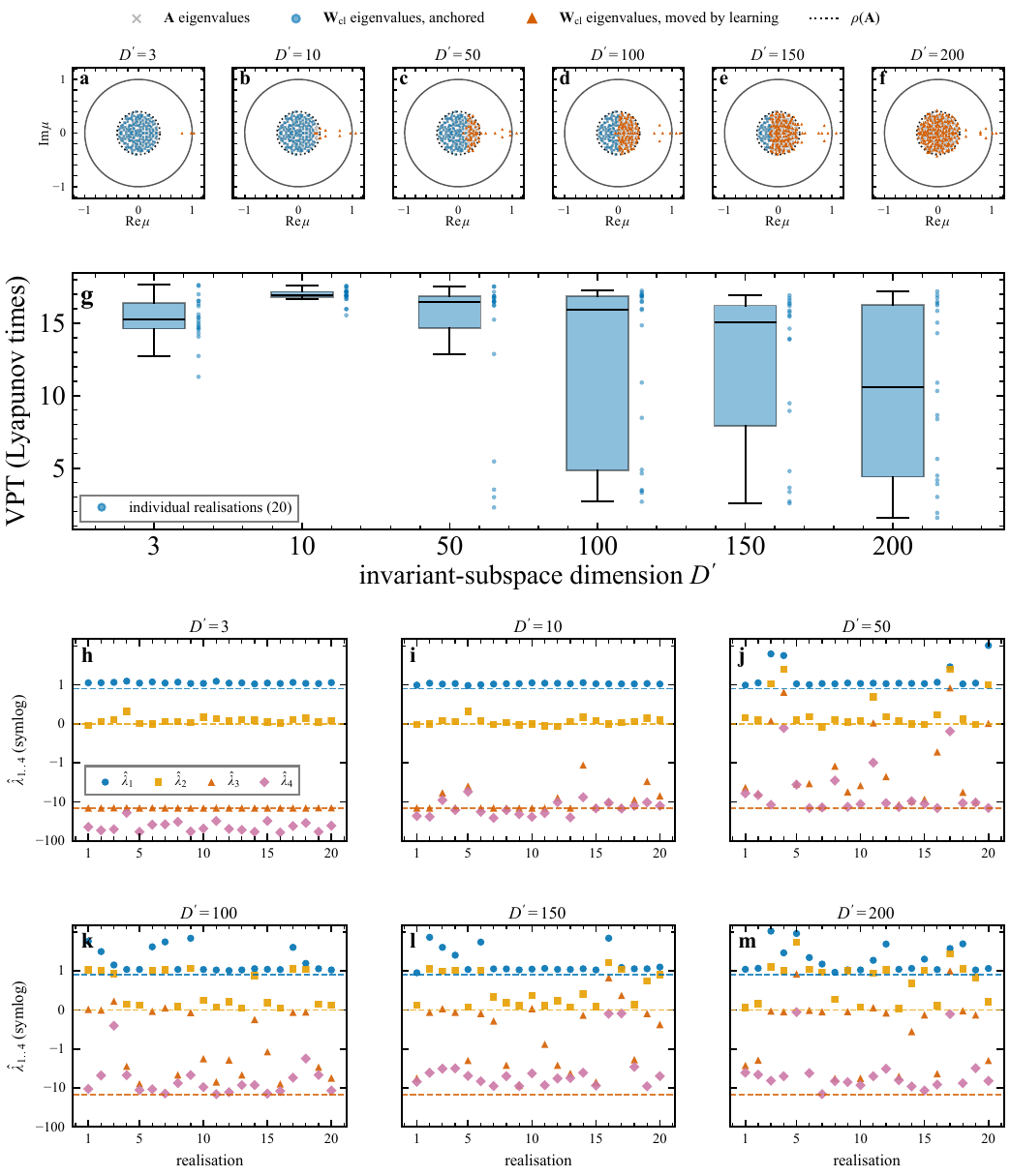}
    \caption{\textbf{Spurious slow modes appear as \(D'\) increases.}
    Analytic construction at the representative setting (Lorenz, \(N=200\), \(\rho=0.4\)), with \(\mathbf{W}_{\mathrm{in}}\) drawn as a random combination within the \(D'\)-dimensional invariant subspace built from the eigenvalues with the largest real parts.
    \textbf{a}--\textbf{f}, Closed-loop spectra for one realization at \(D'=3,10,50,100,150,200\): gray dashed circle, the reservoir spectrum \(\sigma(\mathbf{A})\); blue dots, the anchored eigenvalues; orange triangles, the eigenvalues moved by training. At every \(D'\), the anchored set is the one predicted by Theorem~\ref{thm:invariant_eigenvalues} and the number of moved eigenvalues equals \(D'\).
    \textbf{g}, VPT over 20 reservoir realizations at each \(D'\).
    \textbf{h}--\textbf{m}, Estimated Lyapunov exponents \(\hat\lambda_1\)--\(\hat\lambda_4\) of every realization, on a symmetric-log axis; dashed lines mark the true values. All LEs are recovered by every realization at \(D'=3\), and spurious exponents gradually appear as \(D'\) grows.}
    \label{fig:si-dprime}
\end{figure}

\subsection{Hyperparameter grid}\label{si:grid}

We swept the reservoir topology (Erd\H{o}s--R\'enyi with \(p_A=0.3\), dense random, and dense symmetric, the dense variants with entries drawn from a uniform distribution), the three input-layer constructions (random, analytic, probe-based), the spectral radius \(\rho\in\{0.01, 0.1, 0.2, \dots, 0.9\}\), and the reservoir dimension \(N\in\{50, 60, \dots, 100, 150, 200, 250, 300\}\), with 20 reservoir realizations per cell and \(\beta=0\) throughout (18{,}000 realizations).
For a symmetric \(\mathbf{A}\), whose spectrum is real, both proposed constructions use three real eigenvectors.

The proposed constructions roughly double the region of hyperparameters that sustains accurate prediction  (Supplementary Fig.~\ref{fig:si-grid}).
Cells with mean VPT above \(10\,\lambda_{1}^{-1}\) (orange outline in Supplementary Fig.~\ref{fig:si-grid}) cover \(47\%\) and \(52\%\) of the grid for the analytic and probe-based constructions, against \(25\%\) for the random one.
The Lyapunov spectrum separates the constructions more sharply than the prediction horizon (Supplementary Fig.~\ref{fig:si-grid-le}): the median LSE is \(0.25\)--\(0.28\) for the proposed constructions against \(2.8\)--\(7.5\) for the random one, and the fraction of the grid with LSE below \(0.5\) is \(0.70\)--\(0.83\) against \(0.09\)--\(0.19\).

The reservoir topology, by contrast, has little effect on the performance.
The relation between the topology and RC performance has been studied repeatedly, and Jaeger confirmed the insensitivity \cite{Busing2010-mq}.
Our grid shows the same insensitivity, with the symmetric reservoir slightly better than the other two.
For a symmetric \(\mathbf{A}\), the largest singular value equals the spectral radius, so \(\rho<1\) guarantees \(\lVert\mathbf{A}\rVert_{2}<1\) and the echo state property holds rigorously \cite{Jaeger2001-vs}.
We attribute the slight advantage to the stronger and more uniform contraction that this bound provides.

\begin{figure}[H]
    \centering
    \includegraphics[width=0.99\linewidth,height=\textheight,keepaspectratio]{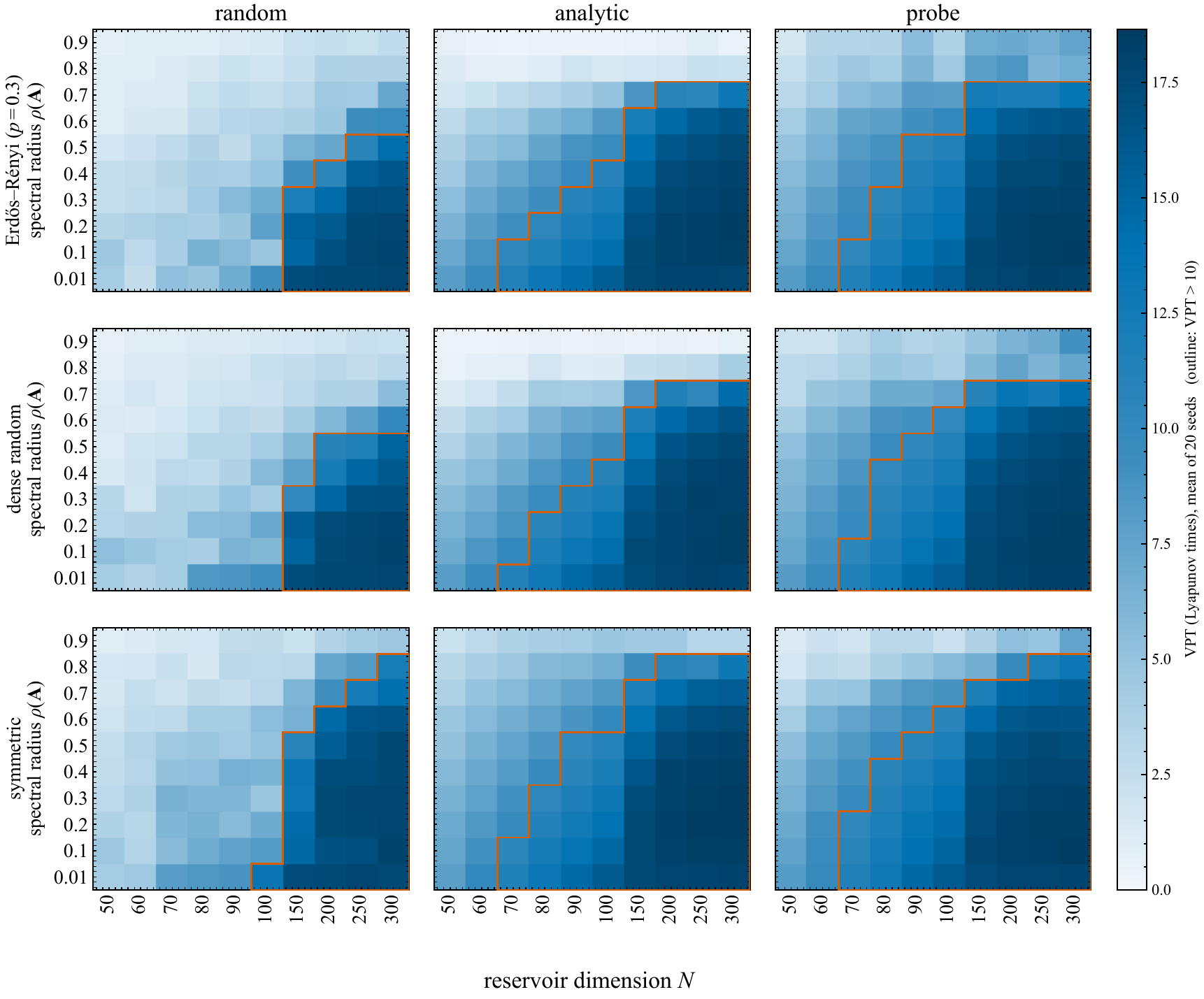}
    \caption{\textbf{The design guidelines hold across the hyperparameter grid.}
    Mean VPT for the Lorenz system (20 reservoir realizations per cell) as a function of the spectral radius \(\rho\) (vertical) and the reservoir dimension \(N\) (horizontal). Rows: reservoir topology; columns: input-layer construction; all nine panels share one color scale.
    The orange outline encloses the cells with mean VPT above \(10\,\lambda_{1}^{-1}\), traced along the sampled cell edges. The outlined region covers \(47\%\) and \(52\%\) of the grid for the analytic and probe-based constructions, against \(25\%\) for the random one.}
    \label{fig:si-grid}
\end{figure}

\begin{figure}[H]
    \centering
    \includegraphics[width=0.99\linewidth,height=\textheight,keepaspectratio]{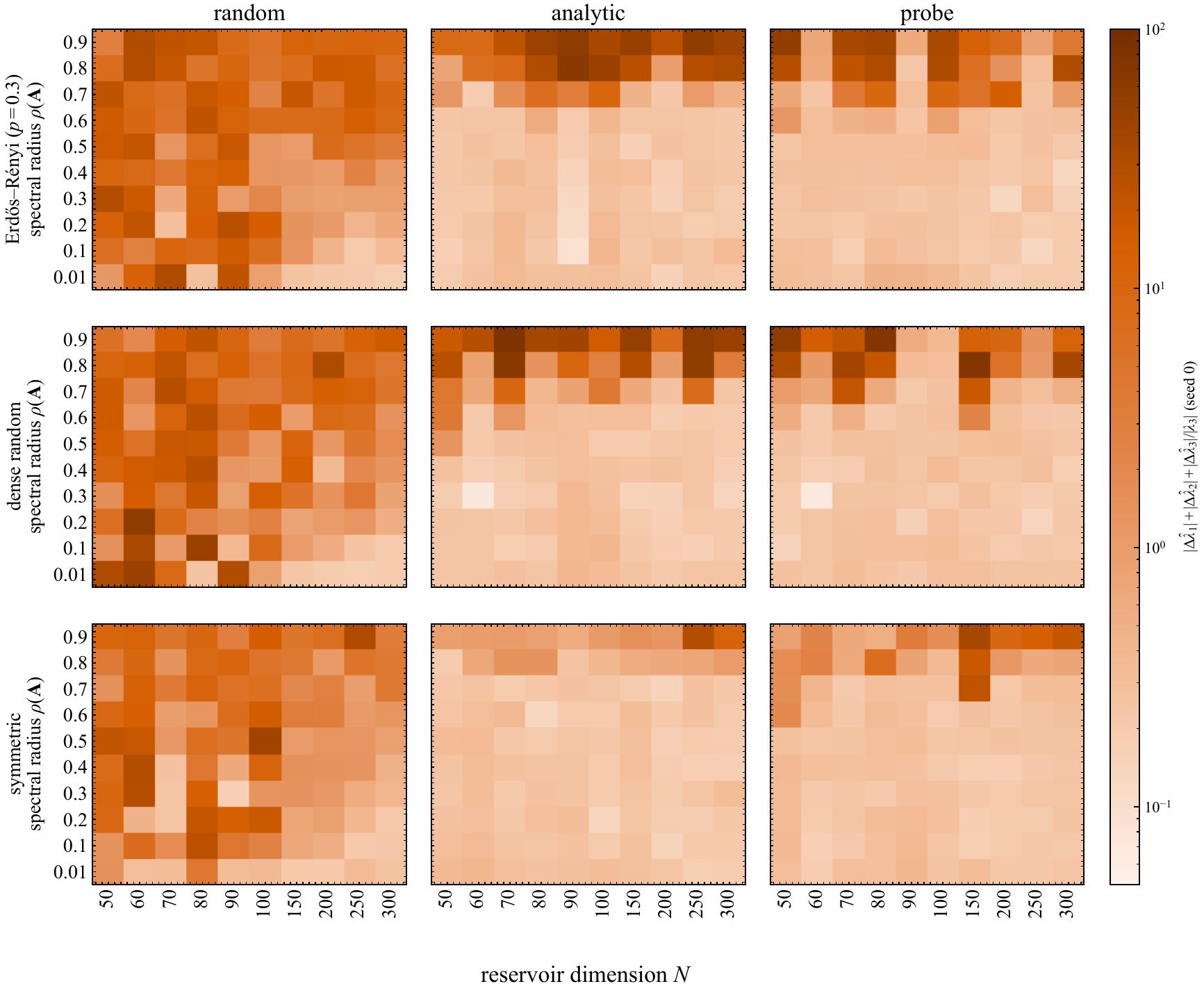}
    \caption{\textbf{The Lyapunov spectrum separates the constructions more sharply than the prediction horizon.}
    Same grid as Supplementary Fig.~\ref{fig:si-grid}, showing the Lyapunov spectrum error (LSE; Methods~\ref{methods-lyapunov}) for one realization, on a logarithmic color scale. The median over each panel is \(0.25\)--\(0.28\) for the proposed constructions against \(2.8\)--\(7.5\) for the random one, and the fraction of the grid with LSE below \(0.5\) is \(0.70\)--\(0.83\) against \(0.09\)--\(0.19\).}
    \label{fig:si-grid-le}
\end{figure}

\subsection{Dependence on the spectral radius}
\label{si-radius-sweep}

In the proposed constructions, the spectral radius \(\rho\) of \(\mathbf{A}\) sets the strength of the fast modes of the closed-loop RC.
That is, it sets the strength of the transverse contraction of the reconstructed attractor.
We therefore examine the dependence on \(\rho\) of the metrics reported in the Results (the distance to the designed subspace, basin stability, domination, VPT, and the Lyapunov spectrum error) in Extended Data Figs.~\ref{fig:ed-sweep-dist}--\ref{fig:ed-sweep-vpt}.

As \(\rho\) grows, the transverse contraction weakens and all metrics eventually degrade.
The distance to the designed subspace grows for every construction, while staying one to two orders of magnitude smaller for the proposed constructions over most of the range (Extended Data Fig.~\ref{fig:ed-sweep-dist}).
Basin stability follows the same pattern (Extended Data Fig.~\ref{fig:ed-sweep-basin}).
The random construction begins losing basins near \(\rho=0.3\) and has lost essentially all of them by \(\rho=0.5\), whereas the median basin of both proposed constructions stays full up to \(\rho=0.5\).
The dominated-splitting test bounds the guaranteed regime (Extended Data Fig.~\ref{fig:ed-sweep-dom}).
Every realization of both proposed constructions is dominated at every sampled point up to \(\rho=0.4\), and domination erodes beyond, whereas the random construction is dominated only at the smallest radius examined.
Prediction tracks these stability margins (Extended Data Fig.~\ref{fig:ed-sweep-vpt}).
The VPT of the random construction declines from moderate \(\rho\) on, while both proposed constructions remain accurate and tightly clustered up to \(\rho=0.6\).
The LSE (Eq.~\eqref{eq:lse}) stays low for the proposed constructions up to \(\rho=0.6\) and grows from the smallest radii for the random one.
At the largest radii the two proposed constructions separate, with the analytic error rising sharply and the probe-based error remaining moderate.

\subsection{Benchmark evaluation details}
\label{si-benchmark-details}

The benchmark of Gilpin evaluates 24 distinct forecasting models, ranging from classical statistical methods to modern deep learning architectures, including a standard RC (Appendix A of ref.~\cite{Gilpin2023-ua}).
Our evaluation (Extended Data Table~\ref{table:dysts}) differs from this reference protocol in a few settings.
Whereas the benchmark evaluates each model once, we averaged the metrics over 10 realizations of the reservoir configuration.
We did not optimize the time-scale parameter (the leaking rate \cite{Jaeger2007-kc}) for each target system, since preliminary experiments showed that the performance of the proposed configuration was relatively insensitive to this parameter.
The excluded system (PanXuZhou) is, at the benchmark's sampling settings, a decaying transient toward a fixed point rather than a sustained chaotic attractor.
Such a simple trajectory is easy to predict and yields an excessively long prediction horizon, so its exclusion does not favor the proposed constructions.

\newpage

\clearpage
\section{Extended Data}
 \label{extended-data}

\renewcommand{\figurename}{Extended Data Fig.}
\renewcommand{\thefigure}{\arabic{figure}}
\setcounter{figure}{0}
\renewcommand{\tablename}{Extended Data Table}
\renewcommand{\thetable}{\arabic{table}}
\setcounter{table}{0}

\begin{table}[h!]
    \centering
    \caption{\textbf{List of hyperparameters and experimental settings.}
        Values denote the specific settings used for the results presented in the corresponding figures and table.
    }\label{table:para}
    \renewcommand{\arraystretch}{1.2}
    \begingroup
    \footnotesize
    \begin{tabular*}{\linewidth}{@{\extracolsep{\fill}}llccc@{}}
        \toprule
        \textbf{Parameter} & \textbf{Symbol} &
        \begin{tabular}[t]{@{}c@{}}\textbf{Figs. \ref{fig:bs}--\ref{fig:vptls}}\end{tabular} &
        \begin{tabular}[t]{@{}c@{}}\textbf{Suppl. Figs. \ref{fig:si-eigsel}--\ref{fig:si-grid-le}}\end{tabular} &
        \textbf{Table \ref{table:dysts}} \\
        \midrule

        Reservoir size & \(N\) & 200 & 200 (S3--S4: 50--300) & 500 \\
        Spectral radius & \(\rho\) & 0.4 (sweeps: 0.01--0.9) & 0.3 (S3--S4: 0.01--0.9) & 0.1 \\
        Reservoir topology & -- & ER & ER & symmetric \\
        Connection probability & \(p_A\) & 0.3 & 0.3 & -- \\
        Regularization & \(\beta\) & 0 & 0 & 0 \\
        Bias vector & \(\mathbf{b}\) & const, \(0.01\) & const, \(0.01\) & const, \(0.01\) \\
        Warmup steps & \(T_{\mathrm{warm}}\) & 1000 & 1000 & 100 \\
        Training steps & \(T_{\mathrm{train}}\) & 2000 & 2000 & 900 \\
        Realizations & -- & 20 & 20 & 10 \\
        Test trials & \(N_{\mathrm{test}}\) & 100 & 100 & 100 \\

        \bottomrule
    \end{tabular*}
    \endgroup
\end{table}

\begin{table}[h!]
    \centering
    \caption{\textbf{Prediction performance across 134 chaotic systems of the \texttt{dysts} benchmark.}
        SMAPE prediction horizon and valid prediction time (VPT), in Lyapunov times (\(\lambda_{1}^{-1}\)), for the reference deep-learning models \cite{Gilpin2023-ua} and for RC with the three input-layer constructions (mean \(\pm\) s.d.\ across systems); all RC results use no regularization (\(\beta=0\)).}\label{table:dysts}

    \begin{tabular*}{\textwidth}{@{\extracolsep{\fill}}lcc@{}}
        \toprule
        \textbf{Model} & \textbf{SMAPE horizon} & \textbf{VPT} \\
        \midrule
        NBEATS \cite{Oreshkin2020-av} / NHiTS \cite{Challu2022-tr} (best of 24 models \cite{Gilpin2023-ua}) & 13.9 \(\pm\) 2.8 & -- \\
        RC, analytic construction & 13.9 \(\pm\) 13.9 & 9.6 \(\pm\) 10.1 \\
        RC, probe-based construction & \textbf{14.1 \(\pm\) 14.0} & \textbf{9.7 \(\pm\) 10.1} \\
        RC, random construction & 6.8 \(\pm\) 10.1 & 4.7 \(\pm\) 7.0 \\
        \bottomrule
    \end{tabular*}
\end{table}

\begin{figure}[H]
    \centering
    \includegraphics[width=0.99\linewidth,height=\textheight,keepaspectratio]{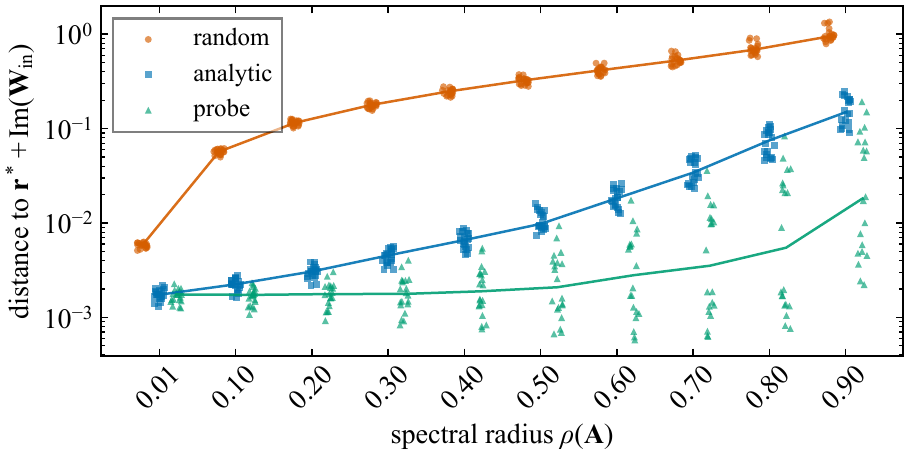}
    \caption{\textbf{Distance to the designed subspace across the spectral radius.} Mean distance from the reconstruction to \(\mathbf{r}^{*}+\mathrm{Im}(\mathbf{W}_{\mathrm{in}})\) as a function of the spectral radius \(\rho\) of \(\mathbf{A}\) (20 reservoir realizations per radius, logarithmic scale). The distance grows with \(\rho\) for every construction, and over most of the range it stays one to two orders of magnitude smaller for the proposed constructions than for the random one. }\label{fig:ed-sweep-dist}
\end{figure}

\begin{figure}[H]
    \centering
    \includegraphics[width=0.99\linewidth,height=\textheight,keepaspectratio]{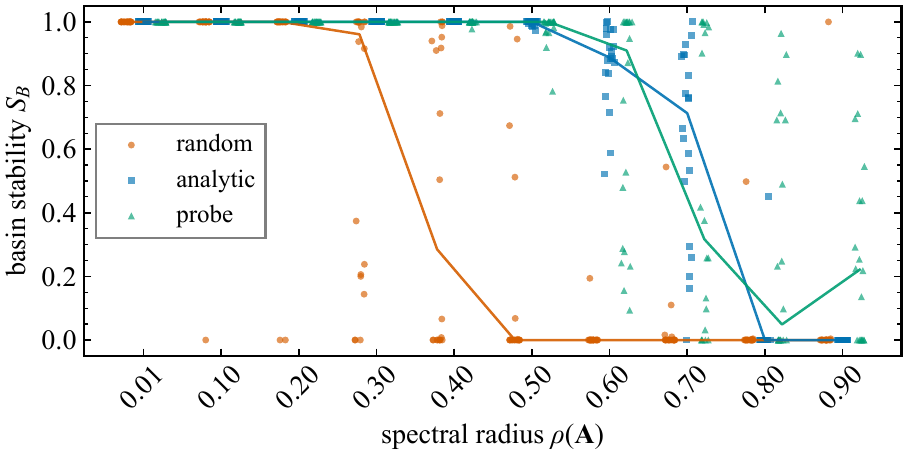}
    \caption{\textbf{Basin stability across the spectral radius.} \(S_B\) as a function of \(\rho\), in the same format as Extended Data Fig.~\ref{fig:ed-sweep-dist}. The random construction degrades from \(\rho\approx0.3\) and loses essentially all trials by \(\rho=0.5\), whereas the median basin of both proposed constructions stays full up to \(\rho=0.5\). }\label{fig:ed-sweep-basin}
\end{figure}

\begin{figure}[H]
    \centering
    \includegraphics[width=0.99\linewidth,height=\textheight,keepaspectratio]{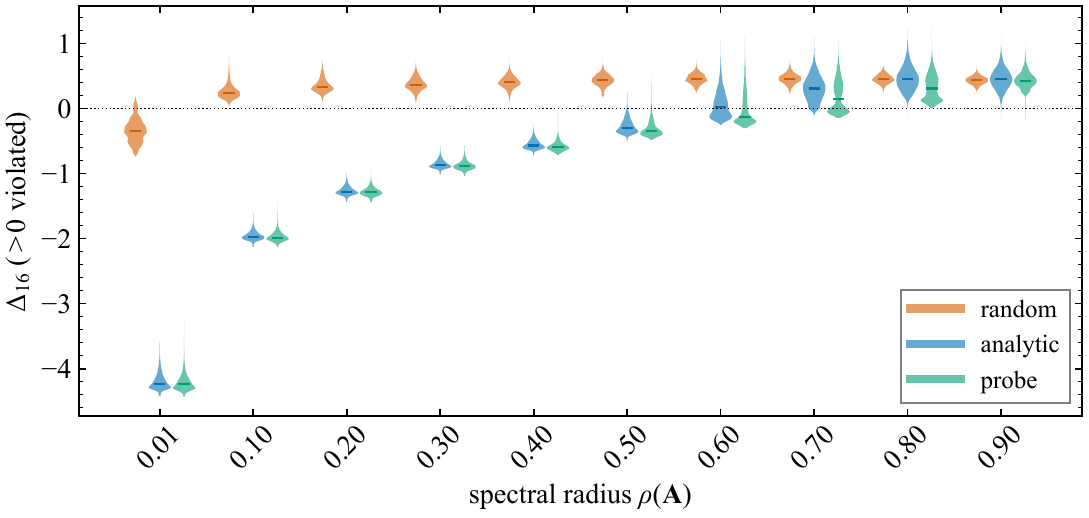}
    \caption{\textbf{Dominated splitting across the spectral radius.} The domination metric \(\Delta_{16}\) along the reconstruction as a function of \(\rho\) (20 reservoir realizations per radius); values above zero violate domination. Up to \(\rho=0.4\), every realization of both proposed constructions is dominated at every sampled point, and domination is gradually lost beyond. The random construction is dominated at every point only at the smallest radius examined. }\label{fig:ed-sweep-dom}
\end{figure}

\begin{figure}[H]
    \centering
    \includegraphics[width=0.99\linewidth,height=\textheight,keepaspectratio]{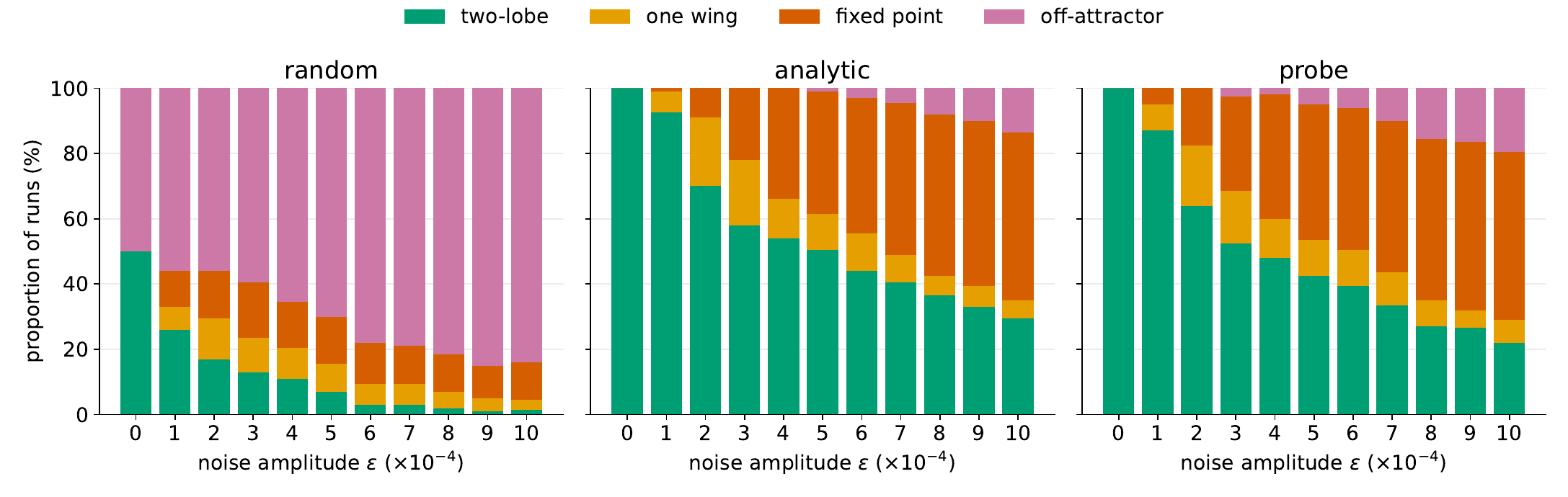}
    \caption{\textbf{Outcome frequencies of the persistence test.} Fraction of runs in each outcome category as a function of the perturbation strength \(\varepsilon\), for the three constructions; each run is classified as two-lobe (the attractor is retained), one wing, fixed point, or off-attractor (Methods~\ref{methods-persistence}). Each bar aggregates 20 reservoir realizations \(\times\) 10 draws of \(\mathbf{W}_{\varepsilon}\), and the leftmost bar is the unperturbed control (20 runs). For the random construction, 10 of the 20 realizations fail already in the control, so its bars mix the perturbation response with baseline failures. }\label{fig:ed-eps-bars}
\end{figure}

\begin{figure}[H]
    \centering
    \includegraphics[width=0.99\linewidth,height=\textheight,keepaspectratio]{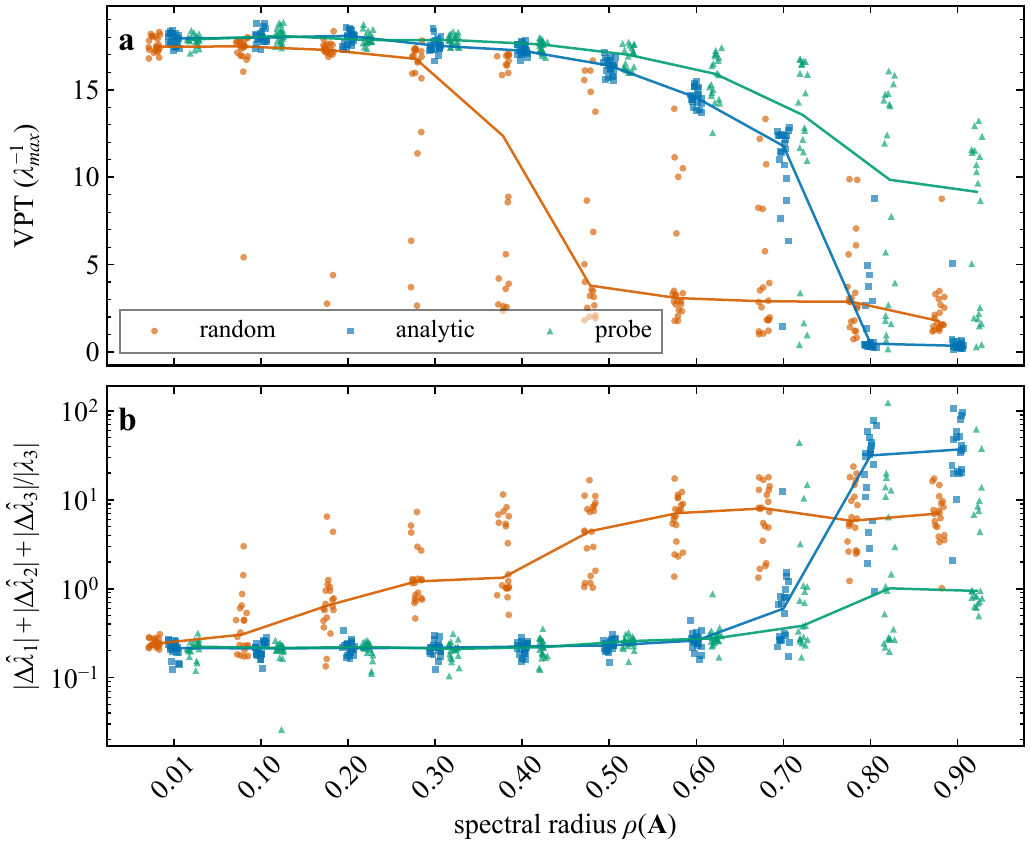}
    \caption{\textbf{Valid prediction time and Lyapunov spectrum error across the spectral radius.} \textbf{a}, VPT as a function of \(\rho\) (20 reservoir realizations per radius). The random construction declines already at moderate \(\rho\) and scatters widely, whereas both proposed constructions remain accurate and tightly clustered up to \(\rho=0.5\). \textbf{b}, Lyapunov spectrum error (LSE; Methods~\ref{methods-lyapunov}) in the same format, on a logarithmic scale. Both proposed constructions hold the error at a low level up to \(\rho=0.6\), whereas the random error grows from the smallest radii. At larger radii the two proposed constructions separate, the analytic error rising sharply while the probe-based error remains moderate. }\label{fig:ed-sweep-vpt}
\end{figure}

\begin{figure}[H]
    \centering
    \includegraphics[width=0.99\linewidth,height=\textheight,keepaspectratio]{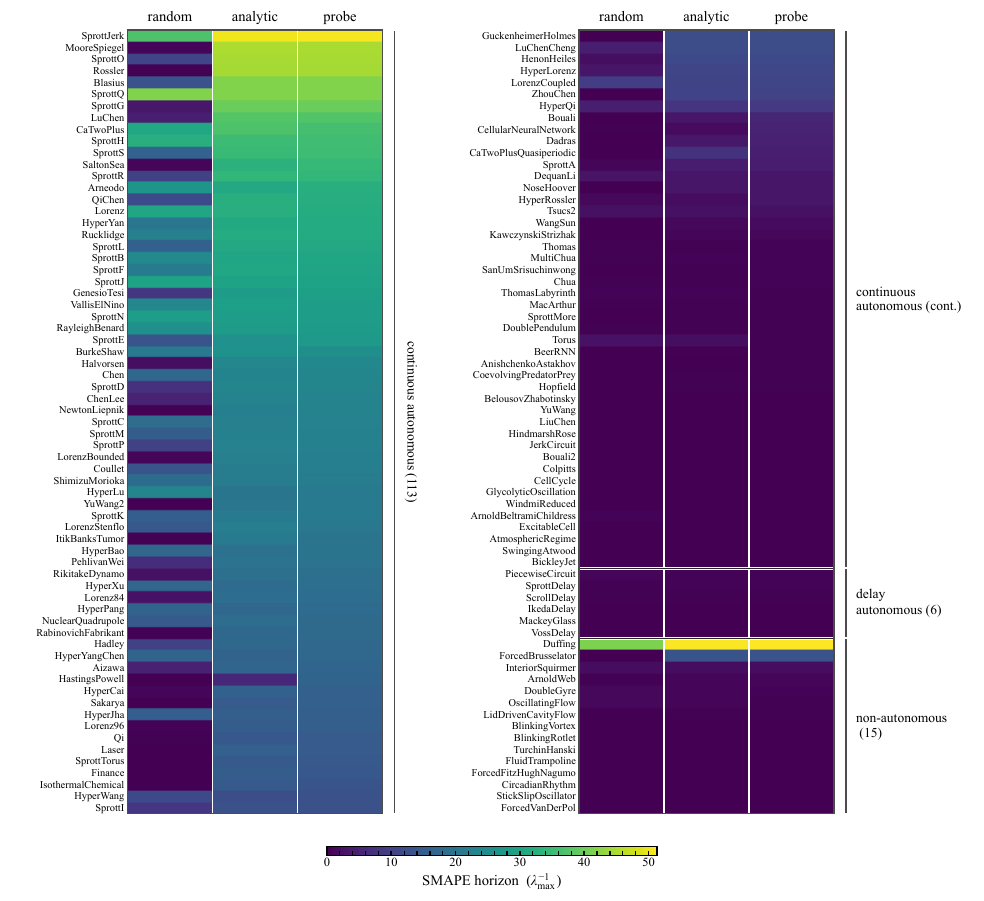}
    \caption{\textbf{SMAPE prediction horizon for each of the 134 benchmark systems.} SMAPE prediction horizon of each system, in units of \(\lambda_{1}^{-1}\) (color scale), for the three input-layer constructions (columns: random, analytic, and probe-based), averaged over 10 reservoir realizations per system. Rows are grouped by system class as in Extended Data Fig.~\ref{fig:si-heatmap} and sorted within each class by the SMAPE horizon of the probe-based construction. }
    \label{fig:ed-benchmark-smape}
\end{figure}

\begin{figure}[H]
    \centering
    \includegraphics[width=0.99\linewidth,height=\textheight,keepaspectratio]{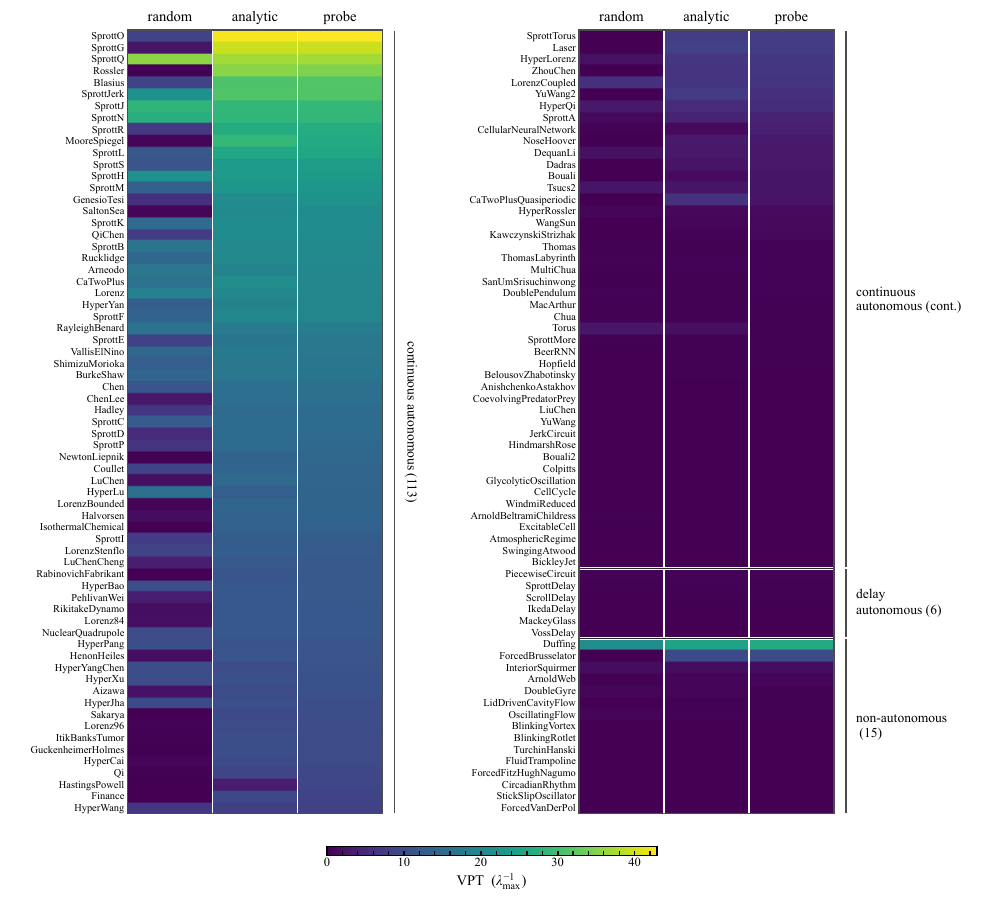}
    \caption{\textbf{Valid prediction time for each of the 134 benchmark systems.}
    VPT of each system, in units of \(\lambda_{1}^{-1}\) (color scale), for the three input-layer constructions (columns: random, analytic, and probe-based), averaged over 10 reservoir realizations per system. Rows are grouped by system class---113 continuous autonomous (continuing in the right panel), 6 delayed autonomous, and 15 non-autonomous---and sorted within each class by the VPT of the probe-based construction. Prediction performance splits markedly along this classification. Many continuous autonomous systems are predicted well, all delayed autonomous systems fail for every construction, and some non-autonomous systems are predicted well.}
    \label{fig:si-heatmap}
\end{figure}

\begin{figure}[H]
    \centering
    \includegraphics[width=0.99\linewidth,height=\textheight,keepaspectratio]{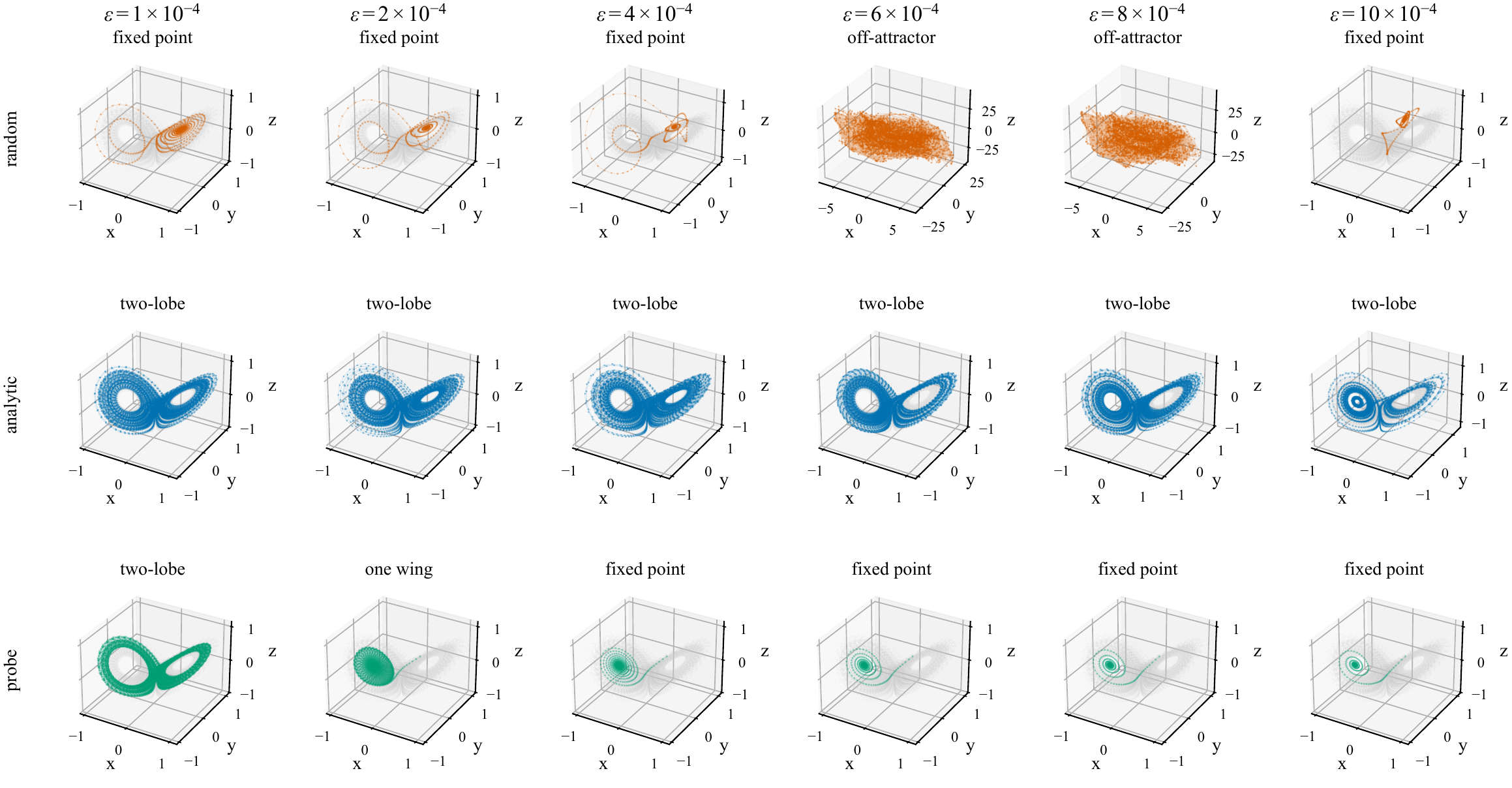}
    \caption{\textbf{Perturbed reconstructions across \(\varepsilon\) for one configuration.} The ladder shows one fixed reservoir realization with one draw of \(\mathbf{W}_{\varepsilon}\) at \(\varepsilon\in\{1,2,4,6,8,10\}\times10^{-4}\). How the attractor breaks down as \(\varepsilon\) grows differs across realizations and draws, and the ensemble trend is given by Extended Data Fig.~\ref{fig:ed-eps-bars}. }\label{fig:ed-eps-ladder}
\end{figure}

\bibliography{paperpile}


\begin{thebibliography}{51}
\ifx \bisbn   \undefined \def \bisbn  #1{ISBN #1}\fi
\ifx \binits  \undefined \def \binits#1{#1}\fi
\ifx \bauthor  \undefined \def \bauthor#1{#1}\fi
\ifx \batitle  \undefined \def \batitle#1{#1}\fi
\ifx \bjtitle  \undefined \def \bjtitle#1{#1}\fi
\ifx \bvolume  \undefined \def \bvolume#1{\textbf{#1}}\fi
\ifx \byear  \undefined \def \byear#1{#1}\fi
\ifx \bissue  \undefined \def \bissue#1{#1}\fi
\ifx \bfpage  \undefined \def \bfpage#1{#1}\fi
\ifx \blpage  \undefined \def \blpage #1{#1}\fi
\ifx \burl  \undefined \def \burl#1{\textsf{#1}}\fi
\ifx \doiurl  \undefined \def \doiurl#1{\url{https://doi.org/#1}}\fi
\ifx \betal  \undefined \def \betal{\textit{et al.}}\fi
\ifx \binstitute  \undefined \def \binstitute#1{#1}\fi
\ifx \binstitutionaled  \undefined \def \binstitutionaled#1{#1}\fi
\ifx \bctitle  \undefined \def \bctitle#1{#1}\fi
\ifx \beditor  \undefined \def \beditor#1{#1}\fi
\ifx \bpublisher  \undefined \def \bpublisher#1{#1}\fi
\ifx \bbtitle  \undefined \def \bbtitle#1{#1}\fi
\ifx \bedition  \undefined \def \bedition#1{#1}\fi
\ifx \bseriesno  \undefined \def \bseriesno#1{#1}\fi
\ifx \blocation  \undefined \def \blocation#1{#1}\fi
\ifx \bsertitle  \undefined \def \bsertitle#1{#1}\fi
\ifx \bsnm \undefined \def \bsnm#1{#1}\fi
\ifx \bsuffix \undefined \def \bsuffix#1{#1}\fi
\ifx \bparticle \undefined \def \bparticle#1{#1}\fi
\ifx \barticle \undefined \def \barticle#1{#1}\fi
\bibcommenthead
\ifx \bconfdate \undefined \def \bconfdate #1{#1}\fi
\ifx \botherref \undefined \def \botherref #1{#1}\fi
\ifx \url \undefined \def \url#1{\textsf{#1}}\fi
\ifx \bchapter \undefined \def \bchapter#1{#1}\fi
\ifx \bbook \undefined \def \bbook#1{#1}\fi
\ifx \bcomment \undefined \def \bcomment#1{#1}\fi
\ifx \oauthor \undefined \def \oauthor#1{#1}\fi
\ifx \citeauthoryear \undefined \def \citeauthoryear#1{#1}\fi
\ifx \endbibitem  \undefined \def \endbibitem {}\fi
\ifx \bconflocation  \undefined \def \bconflocation#1{#1}\fi
\ifx \arxivurl  \undefined \def \arxivurl#1{\textsf{#1}}\fi
\csname PreBibitemsHook\endcsname

\bibitem[\protect\citeauthoryear{Ghil and Lucarini}{2020}]{Ghil2020-jw}
\begin{barticle}
\bauthor{\bsnm{Ghil}, \binits{M.}},
\bauthor{\bsnm{Lucarini}, \binits{V.}}:
\batitle{The physics of climate variability and climate change}.
\bjtitle{Rev. Mod. Phys.}
\bvolume{92}(\bissue{3}),
\bfpage{035002}
(\byear{2020})
\end{barticle}
\endbibitem

\bibitem[\protect\citeauthoryear{Vyas et~al.}{2020}]{Vyas2020-ie}
\begin{barticle}
\bauthor{\bsnm{Vyas}, \binits{S.}},
\bauthor{\bsnm{Golub}, \binits{M.D.}},
\bauthor{\bsnm{Sussillo}, \binits{D.}},
\bauthor{\bsnm{Shenoy}, \binits{K.V.}}:
\batitle{Computation through neural population dynamics}.
\bjtitle{Annu. Rev. Neurosci.}
\bvolume{43}(\bissue{1}),
\bfpage{249}--\blpage{275}
(\byear{2020})
\end{barticle}
\endbibitem

\bibitem[\protect\citeauthoryear{Takens}{1981}]{Takens1981-mw}
\begin{bchapter}
\bauthor{\bsnm{Takens}, \binits{F.}}:
\bctitle{Detecting strange attractors in turbulence}.
In: \beditor{\bsnm{Rand}, \binits{D.}},
\beditor{\bsnm{Young}, \binits{L.-S.}} (eds.)
\bbtitle{Dynamical Systems and Turbulence, Warwick 1980},
pp. \bfpage{366}--\blpage{381}.
\bpublisher{Springer},
\blocation{Berlin, Heidelberg}
(\byear{1981})
\end{bchapter}
\endbibitem

\bibitem[\protect\citeauthoryear{Sauer et~al.}{1991}]{Sauer1991-ew}
\begin{barticle}
\bauthor{\bsnm{Sauer}, \binits{T.}},
\bauthor{\bsnm{Yorke}, \binits{J.A.}},
\bauthor{\bsnm{Casdagli}, \binits{M.}}:
\batitle{Embedology}.
\bjtitle{J. Stat. Phys.}
\bvolume{65}(\bissue{3}),
\bfpage{579}--\blpage{616}
(\byear{1991})
\end{barticle}
\endbibitem

\bibitem[\protect\citeauthoryear{Mikhaeil et~al.}{2021}]{Mikhaeil2021-ve}
\begin{barticle}
\bauthor{\bsnm{Mikhaeil}, \binits{J.M.}},
\bauthor{\bsnm{Monfared}, \binits{Z.}},
\bauthor{\bsnm{Durstewitz}, \binits{D.}}:
\batitle{On the difficulty of learning chaotic dynamics with {RNNs}}.
\bjtitle{Neural Inf Process Syst}
\bvolume{35},
\bfpage{11297}--\blpage{11312}
(\byear{2021})
\end{barticle}
\endbibitem

\bibitem[\protect\citeauthoryear{Jaeger}{2001}]{Jaeger2001-vs}
\begin{botherref}
\oauthor{\bsnm{Jaeger}, \binits{H.}}:
The ``echo state'' approach to analysing and training recurrent neural
  networks.
GMD Technical Report
(148)
(2001)
\end{botherref}
\endbibitem

\bibitem[\protect\citeauthoryear{Maass et~al.}{2002}]{Maass2002-me}
\begin{barticle}
\bauthor{\bsnm{Maass}, \binits{W.}},
\bauthor{\bsnm{Natschläger}, \binits{T.}},
\bauthor{\bsnm{Markram}, \binits{H.}}:
\batitle{Real-time computing without stable states: a new framework for neural
  computation based on perturbations}.
\bjtitle{Neural Comput.}
\bvolume{14}(\bissue{11}),
\bfpage{2531}--\blpage{2560}
(\byear{2002})
\end{barticle}
\endbibitem

\bibitem[\protect\citeauthoryear{Tanaka et~al.}{2019}]{Tanaka2019-vn}
\begin{barticle}
\bauthor{\bsnm{Tanaka}, \binits{G.}},
\bauthor{\bsnm{Yamane}, \binits{T.}},
\bauthor{\bsnm{Héroux}, \binits{J.B.}},
\bauthor{\bsnm{Nakane}, \binits{R.}},
\bauthor{\bsnm{Kanazawa}, \binits{N.}},
\bauthor{\bsnm{Takeda}, \binits{S.}},
\bauthor{\bsnm{Numata}, \binits{H.}},
\bauthor{\bsnm{Nakano}, \binits{D.}},
\bauthor{\bsnm{Hirose}, \binits{A.}}:
\batitle{Recent advances in physical reservoir computing: A review}.
\bjtitle{Neural Netw.}
\bvolume{115},
\bfpage{100}--\blpage{123}
(\byear{2019})
\end{barticle}
\endbibitem

\bibitem[\protect\citeauthoryear{Pathak et~al.}{2017}]{Pathak2017-id}
\begin{barticle}
\bauthor{\bsnm{Pathak}, \binits{J.}},
\bauthor{\bsnm{Lu}, \binits{Z.}},
\bauthor{\bsnm{Hunt}, \binits{B.R.}},
\bauthor{\bsnm{Girvan}, \binits{M.}},
\bauthor{\bsnm{Ott}, \binits{E.}}:
\batitle{Using machine learning to replicate chaotic attractors and calculate
  lyapunov exponents from data}.
\bjtitle{Chaos}
\bvolume{27}(\bissue{12}),
\bfpage{121102}
(\byear{2017})
\end{barticle}
\endbibitem

\bibitem[\protect\citeauthoryear{Lu et~al.}{2018}]{Lu2018-vk}
\begin{barticle}
\bauthor{\bsnm{Lu}, \binits{Z.}},
\bauthor{\bsnm{Hunt}, \binits{B.R.}},
\bauthor{\bsnm{Ott}, \binits{E.}}:
\batitle{Attractor reconstruction by machine learning}.
\bjtitle{Chaos}
\bvolume{28}(\bissue{6}),
\bfpage{061104}
(\byear{2018})
\end{barticle}
\endbibitem

\bibitem[\protect\citeauthoryear{Hart et~al.}{2020}]{Hart2020-kx}
\begin{barticle}
\bauthor{\bsnm{Hart}, \binits{A.}},
\bauthor{\bsnm{Hook}, \binits{J.}},
\bauthor{\bsnm{Dawes}, \binits{J.}}:
\batitle{Embedding and approximation theorems for echo state networks}.
\bjtitle{Neural Netw.}
\bvolume{128},
\bfpage{234}--\blpage{247}
(\byear{2020})
\end{barticle}
\endbibitem

\bibitem[\protect\citeauthoryear{Grigoryeva et~al.}{2021}]{Grigoryeva2021-mg}
\begin{barticle}
\bauthor{\bsnm{Grigoryeva}, \binits{L.}},
\bauthor{\bsnm{Hart}, \binits{A.}},
\bauthor{\bsnm{Ortega}, \binits{J.-P.}}:
\batitle{Chaos on compact manifolds: Differentiable synchronizations beyond the
  takens theorem}.
\bjtitle{Phys Rev E}
\bvolume{103}(\bissue{6-1}),
\bfpage{062204}
(\byear{2021})
\end{barticle}
\endbibitem

\bibitem[\protect\citeauthoryear{Margazoglou and
  Magri}{2023}]{Margazoglou2023-lc}
\begin{barticle}
\bauthor{\bsnm{Margazoglou}, \binits{G.}},
\bauthor{\bsnm{Magri}, \binits{L.}}:
\batitle{Stability analysis of chaotic systems from data}.
\bjtitle{Nonlinear Dyn.}
\bvolume{111}(\bissue{9}),
\bfpage{8799}--\blpage{8819}
(\byear{2023})
\end{barticle}
\endbibitem

\bibitem[\protect\citeauthoryear{Sisodia and Jalan}{2024}]{Sisodia2024-fl}
\begin{barticle}
\bauthor{\bsnm{Sisodia}, \binits{D.}},
\bauthor{\bsnm{Jalan}, \binits{S.}}:
\batitle{Dynamical analysis of a parameter-aware reservoir computer}.
\bjtitle{Phys. Rev. E.}
\bvolume{110}(\bissue{3-1}),
\bfpage{034211}
(\byear{2024})
\end{barticle}
\endbibitem

\bibitem[\protect\citeauthoryear{Bi et~al.}{2023}]{Bi2023-jr}
\begin{barticle}
\bauthor{\bsnm{Bi}, \binits{K.}},
\bauthor{\bsnm{Xie}, \binits{L.}},
\bauthor{\bsnm{Zhang}, \binits{H.}},
\bauthor{\bsnm{Chen}, \binits{X.}},
\bauthor{\bsnm{Gu}, \binits{X.}},
\bauthor{\bsnm{Tian}, \binits{Q.}}:
\batitle{Accurate medium-range global weather forecasting with {3D} neural
  networks}.
\bjtitle{Nature}
\bvolume{619}(\bissue{7970}),
\bfpage{533}--\blpage{538}
(\byear{2023})
\end{barticle}
\endbibitem

\bibitem[\protect\citeauthoryear{Kochkov et~al.}{2024}]{Kochkov2024-ak}
\begin{barticle}
\bauthor{\bsnm{Kochkov}, \binits{D.}},
\bauthor{\bsnm{Yuval}, \binits{J.}},
\bauthor{\bsnm{Langmore}, \binits{I.}},
\bauthor{\bsnm{Norgaard}, \binits{P.}},
\bauthor{\bsnm{Smith}, \binits{J.}},
\bauthor{\bsnm{Mooers}, \binits{G.}},
\bauthor{\bsnm{Klöwer}, \binits{M.}},
\bauthor{\bsnm{Lottes}, \binits{J.}},
\bauthor{\bsnm{Rasp}, \binits{S.}},
\bauthor{\bsnm{Düben}, \binits{P.}},
\bauthor{\bsnm{Hatfield}, \binits{S.}},
\bauthor{\bsnm{Battaglia}, \binits{P.}},
\bauthor{\bsnm{Sanchez-Gonzalez}, \binits{A.}},
\bauthor{\bsnm{Willson}, \binits{M.}},
\bauthor{\bsnm{Brenner}, \binits{M.P.}},
\bauthor{\bsnm{Hoyer}, \binits{S.}}:
\batitle{Neural general circulation models for weather and climate}.
\bjtitle{Nature}
\bvolume{632}(\bissue{8027}),
\bfpage{1060}--\blpage{1066}
(\byear{2024})
\end{barticle}
\endbibitem

\bibitem[\protect\citeauthoryear{Platt et~al.}{2023}]{Platt2023-lh}
\begin{barticle}
\bauthor{\bsnm{Platt}, \binits{J.A.}},
\bauthor{\bsnm{Penny}, \binits{S.G.}},
\bauthor{\bsnm{Smith}, \binits{T.A.}},
\bauthor{\bsnm{Chen}, \binits{T.-C.}},
\bauthor{\bsnm{Abarbanel}, \binits{H.D.I.}}:
\batitle{Constraining chaos: Enforcing dynamical invariants in the training of
  reservoir computers}.
\bjtitle{Chaos}
\bvolume{33}(\bissue{10}),
\bfpage{103107}
(\byear{2023})
\end{barticle}
\endbibitem

\bibitem[\protect\citeauthoryear{Lorenz}{1963}]{Lorenz1963-kf}
\begin{barticle}
\bauthor{\bsnm{Lorenz}, \binits{E.N.}}:
\batitle{Deterministic nonperiodic flow}.
\bjtitle{J. Atmos. Sci.}
\bvolume{20}(\bissue{2}),
\bfpage{130}--\blpage{141}
(\byear{1963})
\end{barticle}
\endbibitem

\bibitem[\protect\citeauthoryear{Menck et~al.}{2013}]{Menck2013-ex}
\begin{barticle}
\bauthor{\bsnm{Menck}, \binits{P.J.}},
\bauthor{\bsnm{Heitzig}, \binits{J.}},
\bauthor{\bsnm{Marwan}, \binits{N.}},
\bauthor{\bsnm{Kurths}, \binits{J.}}:
\batitle{How basin stability complements the linear-stability paradigm}.
\bjtitle{Nat. Phys.}
\bvolume{9}(\bissue{2}),
\bfpage{89}--\blpage{92}
(\byear{2013})
\end{barticle}
\endbibitem

\bibitem[\protect\citeauthoryear{Fenichel}{1979}]{Fenichel1979-ej}
\begin{barticle}
\bauthor{\bsnm{Fenichel}, \binits{N.}}:
\batitle{Geometric singular perturbation theory for ordinary differential
  equations}.
\bjtitle{J. Differ. Equ.}
\bvolume{31}(\bissue{1}),
\bfpage{53}--\blpage{98}
(\byear{1979})
\end{barticle}
\endbibitem

\bibitem[\protect\citeauthoryear{Hirsch et~al.}{1977}]{Hirsch1977-ry}
\begin{bbook}
\bauthor{\bsnm{Hirsch}, \binits{M.W.}},
\bauthor{\bsnm{Pugh}, \binits{C.C.}},
\bauthor{\bsnm{Shub}, \binits{M.}}:
\bbtitle{Invariant Manifolds},
\bedition{1977} edn.
\bsertitle{Lecture notes in mathematics}.
\bpublisher{Springer},
\blocation{Berlin, Germany}
(\byear{1977})
\end{bbook}
\endbibitem

\bibitem[\protect\citeauthoryear{Vlachas et~al.}{2020}]{Vlachas2020-ob}
\begin{barticle}
\bauthor{\bsnm{Vlachas}, \binits{P.R.}},
\bauthor{\bsnm{Pathak}, \binits{J.}},
\bauthor{\bsnm{Hunt}, \binits{B.R.}},
\bauthor{\bsnm{Sapsis}, \binits{T.P.}},
\bauthor{\bsnm{Girvan}, \binits{M.}},
\bauthor{\bsnm{Ott}, \binits{E.}},
\bauthor{\bsnm{Koumoutsakos}, \binits{P.}}:
\batitle{Backpropagation algorithms and reservoir computing in recurrent neural
  networks for the forecasting of complex spatiotemporal dynamics}.
\bjtitle{Neural Netw.}
\bvolume{126},
\bfpage{191}--\blpage{217}
(\byear{2020})
\end{barticle}
\endbibitem

\bibitem[\protect\citeauthoryear{Platt et~al.}{2021}]{Platt2021-el}
\begin{barticle}
\bauthor{\bsnm{Platt}, \binits{J.A.}},
\bauthor{\bsnm{Wong}, \binits{A.}},
\bauthor{\bsnm{Clark}, \binits{R.}},
\bauthor{\bsnm{Penny}, \binits{S.G.}},
\bauthor{\bsnm{Abarbanel}, \binits{H.D.I.}}:
\batitle{Robust forecasting using predictive generalized synchronization in
  reservoir computing}.
\bjtitle{Chaos}
\bvolume{31}(\bissue{12}),
\bfpage{123118}
(\byear{2021})
\end{barticle}
\endbibitem

\bibitem[\protect\citeauthoryear{Bollt}{2021}]{Bollt2021-mj}
\begin{barticle}
\bauthor{\bsnm{Bollt}, \binits{E.}}:
\batitle{On explaining the surprising success of reservoir computing forecaster
  of chaos? the universal machine learning dynamical system with contrast to
  {VAR} and {DMD}}.
\bjtitle{Chaos}
\bvolume{31}(\bissue{1}),
\bfpage{013108}
(\byear{2021})
\end{barticle}
\endbibitem

\bibitem[\protect\citeauthoryear{Hurley and Shaheen}{2025}]{Hurley2025-ue}
\begin{botherref}
\oauthor{\bsnm{Hurley}, \binits{L.A.}},
\oauthor{\bsnm{Shaheen}, \binits{S.E.}}:
Reservoir computing with large valid prediction time for the lorenz system.
arXiv [cs.NE]
(2025)
{[cs.NE]}
\end{botherref}
\endbibitem

\bibitem[\protect\citeauthoryear{Sauer et~al.}{1998}]{Sauer1998-la}
\begin{barticle}
\bauthor{\bsnm{Sauer}, \binits{T.D.}},
\bauthor{\bsnm{Tempkin}, \binits{J.A.}},
\bauthor{\bsnm{Yorke}, \binits{J.A.}}:
\batitle{Spurious lyapunov exponents in attractor reconstruction}.
\bjtitle{Phys. Rev. Lett.}
\bvolume{81}(\bissue{20}),
\bfpage{4341}--\blpage{4344}
(\byear{1998})
\end{barticle}
\endbibitem

\bibitem[\protect\citeauthoryear{Kaplan and Yorke}{1979}]{Kaplan1979-jy}
\begin{bchapter}
\bauthor{\bsnm{Kaplan}, \binits{J.L.}},
\bauthor{\bsnm{Yorke}, \binits{J.A.}}:
\bctitle{Chaotic behavior of multidimensional difference equations}.
In: \bbtitle{Functional Differential Equations and Approximation of Fixed
  Points}.
\bsertitle{Lecture notes in mathematics},
pp. \bfpage{204}--\blpage{227}.
\bpublisher{Springer},
\blocation{Berlin, Heidelberg}
(\byear{1979})
\end{bchapter}
\endbibitem

\bibitem[\protect\citeauthoryear{Gilpin}{2021}]{Gilpin2021-wp}
\begin{botherref}
\oauthor{\bsnm{Gilpin}, \binits{W.}}:
Chaos as an interpretable benchmark for forecasting and data-driven modelling.
arXiv [cs.LG]
(2021)
{[cs.LG]}
\end{botherref}
\endbibitem

\bibitem[\protect\citeauthoryear{Gilpin}{2023}]{Gilpin2023-ua}
\begin{barticle}
\bauthor{\bsnm{Gilpin}, \binits{W.}}:
\batitle{Model scale versus domain knowledge in statistical forecasting of
  chaotic systems}.
\bjtitle{Phys. Rev. Res.}
\bvolume{5}(\bissue{4}),
\bfpage{043252}
(\byear{2023})
\end{barticle}
\endbibitem

\bibitem[\protect\citeauthoryear{Oreshkin et~al.}{2020}]{Oreshkin2020-av}
\begin{bchapter}
\bauthor{\bsnm{Oreshkin}, \binits{B.N.}},
\bauthor{\bsnm{Carpov}, \binits{D.}},
\bauthor{\bsnm{Chapados}, \binits{N.}},
\bauthor{\bsnm{Bengio}, \binits{Y.}}:
\bctitle{{N}-{BEATS}: Neural basis expansion analysis for interpretable time
  series forecasting}.
In: \bbtitle{International Conference on Learning Representations}
(\byear{2020})
\end{bchapter}
\endbibitem

\bibitem[\protect\citeauthoryear{Challu et~al.}{2022}]{Challu2022-tr}
\begin{botherref}
\oauthor{\bsnm{Challu}, \binits{C.}},
\oauthor{\bsnm{Olivares}, \binits{K.G.}},
\oauthor{\bsnm{Oreshkin}, \binits{B.N.}},
\oauthor{\bsnm{Garza}, \binits{F.}},
\oauthor{\bsnm{Mergenthaler-Canseco}, \binits{M.}},
\oauthor{\bsnm{Dubrawski}, \binits{A.}}:
{N}-{HiTS}: Neural hierarchical interpolation for time series forecasting.
National Conference on Artificial Intelligence
\textbf{abs/2201.12886}
(2022)
\end{botherref}
\endbibitem

\bibitem[\protect\citeauthoryear{Kuehn}{2015}]{Kuehn2015-pl}
\begin{bbook}
\bauthor{\bsnm{Kuehn}, \binits{C.}}:
\bbtitle{Multiple Time Scale Dynamics},
\bedition{2015} edn.
\bsertitle{Applied Mathematical Sciences}.
\bpublisher{Springer},
\blocation{Cham, Switzerland}
(\byear{2015})
\end{bbook}
\endbibitem

\bibitem[\protect\citeauthoryear{Kailath}{1980}]{Kailath1980-ou}
\begin{botherref}
\oauthor{\bsnm{Kailath}, \binits{T.}}:
Linear systems
\textbf{156}
(1980)
\end{botherref}
\endbibitem

\bibitem[\protect\citeauthoryear{Grigoryeva et~al.}{2023}]{Grigoryeva2023-cw}
\begin{barticle}
\bauthor{\bsnm{Grigoryeva}, \binits{L.}},
\bauthor{\bsnm{Hart}, \binits{A.}},
\bauthor{\bsnm{Ortega}, \binits{J.-P.}}:
\batitle{Learning strange attractors with reservoir systems}.
\bjtitle{Nonlinearity}
\bvolume{36}(\bissue{9}),
\bfpage{4674}
(\byear{2023})
\end{barticle}
\endbibitem

\bibitem[\protect\citeauthoryear{Liu et~al.}{2011}]{Liu2011-qs}
\begin{barticle}
\bauthor{\bsnm{Liu}, \binits{Y.-Y.}},
\bauthor{\bsnm{Slotine}, \binits{J.-J.}},
\bauthor{\bsnm{Barabási}, \binits{A.-L.}}:
\batitle{Controllability of complex networks}.
\bjtitle{Nature}
\bvolume{473}(\bissue{7346}),
\bfpage{167}--\blpage{173}
(\byear{2011})
\end{barticle}
\endbibitem

\bibitem[\protect\citeauthoryear{Gu et~al.}{2021}]{Gu2021-uh}
\begin{botherref}
\oauthor{\bsnm{Gu}, \binits{A.}},
\oauthor{\bsnm{Goel}, \binits{K.}},
\oauthor{\bsnm{Ré}, \binits{C.}}:
Efficiently modeling long sequences with structured state spaces.
arXiv [cs.LG]
(2021)
{[cs.LG]}
\end{botherref}
\endbibitem

\bibitem[\protect\citeauthoryear{Gu and Dao}{2023}]{Gu2023-xy}
\begin{botherref}
\oauthor{\bsnm{Gu}, \binits{A.}},
\oauthor{\bsnm{Dao}, \binits{T.}}:
Mamba: Linear-time sequence modeling with selective state spaces.
arXiv [cs.LG]
(2023)
{[cs.LG]}
\end{botherref}
\endbibitem

\bibitem[\protect\citeauthoryear{Brunton et~al.}{2016}]{Brunton2016-gc}
\begin{barticle}
\bauthor{\bsnm{Brunton}, \binits{S.L.}},
\bauthor{\bsnm{Proctor}, \binits{J.L.}},
\bauthor{\bsnm{Kutz}, \binits{J.N.}}:
\batitle{Discovering governing equations from data by sparse identification of
  nonlinear dynamical systems}.
\bjtitle{Proc. Natl. Acad. Sci. U. S. A.}
\bvolume{113}(\bissue{15}),
\bfpage{3932}--\blpage{3937}
(\byear{2016})
\end{barticle}
\endbibitem

\bibitem[\protect\citeauthoryear{Schmid}{2010}]{Schmid2010-lj}
\begin{barticle}
\bauthor{\bsnm{Schmid}, \binits{P.J.}}:
\batitle{Dynamic mode decomposition of numerical and experimental data}.
\bjtitle{J. Fluid Mech.}
\bvolume{656},
\bfpage{5}--\blpage{28}
(\byear{2010})
\end{barticle}
\endbibitem

\bibitem[\protect\citeauthoryear{Williams et~al.}{2015}]{Williams2015-jg}
\begin{barticle}
\bauthor{\bsnm{Williams}, \binits{M.O.}},
\bauthor{\bsnm{Kevrekidis}, \binits{I.G.}},
\bauthor{\bsnm{Rowley}, \binits{C.W.}}:
\batitle{A data–driven approximation of the koopman operator: Extending
  dynamic mode decomposition}.
\bjtitle{J. Nonlinear Sci.}
\bvolume{25}(\bissue{6}),
\bfpage{1307}--\blpage{1346}
(\byear{2015})
\end{barticle}
\endbibitem

\bibitem[\protect\citeauthoryear{Jaeger}{2007}]{Jaeger2007-py}
\begin{barticle}
\bauthor{\bsnm{Jaeger}, \binits{H.}}:
\batitle{Echo state network}.
\bjtitle{Scholarpedia}
\bvolume{2}(\bissue{9}),
\bfpage{2330}
(\byear{2007})
\end{barticle}
\endbibitem

\bibitem[\protect\citeauthoryear{Lukoševičius and
  Jaeger}{2009}]{Lukosevicius2009-vt}
\begin{barticle}
\bauthor{\bsnm{Lukoševičius}, \binits{M.}},
\bauthor{\bsnm{Jaeger}, \binits{H.}}:
\batitle{Reservoir computing approaches to recurrent neural network training}.
\bjtitle{Comput. Sci. Rev.}
\bvolume{3}(\bissue{3}),
\bfpage{127}--\blpage{149}
(\byear{2009})
\end{barticle}
\endbibitem

\bibitem[\protect\citeauthoryear{Shimada and Nagashima}{1979}]{Shimada1979-mv}
\begin{barticle}
\bauthor{\bsnm{Shimada}, \binits{I.}},
\bauthor{\bsnm{Nagashima}, \binits{T.}}:
\batitle{A numerical approach to ergodic problem of dissipative dynamical
  systems}.
\bjtitle{Prog. Theor. Phys.}
\bvolume{61}(\bissue{6}),
\bfpage{1605}--\blpage{1616}
(\byear{1979})
\end{barticle}
\endbibitem

\bibitem[\protect\citeauthoryear{Benettin et~al.}{1980}]{Benettin1980-tk}
\begin{barticle}
\bauthor{\bsnm{Benettin}, \binits{G.}},
\bauthor{\bsnm{Galgani}, \binits{L.}},
\bauthor{\bsnm{Giorgilli}, \binits{A.}},
\bauthor{\bsnm{Strelcyn}, \binits{J.-M.}}:
\batitle{Lyapunov characteristic exponents for smooth dynamical systems and for
  hamiltonian systems; a method for computing all of them. part 2: Numerical
  application}.
\bjtitle{Meccanica}
\bvolume{15}(\bissue{1}),
\bfpage{21}--\blpage{30}
(\byear{1980})
\end{barticle}
\endbibitem

\bibitem[\protect\citeauthoryear{Fujisaka and Yamada}{1983}]{Fujisaka1983-hu}
\begin{barticle}
\bauthor{\bsnm{Fujisaka}, \binits{H.}},
\bauthor{\bsnm{Yamada}, \binits{T.}}:
\batitle{Stability theory of synchronized motion in coupled-oscillator
  systems:}.
\bjtitle{Prog Theor Phys}
\bvolume{69}(\bissue{1}),
\bfpage{32}--\blpage{47}
(\byear{1983})
\end{barticle}
\endbibitem

\bibitem[\protect\citeauthoryear{Rulkov et~al.}{1995}]{Rulkov1995-tq}
\begin{barticle}
\bauthor{\bsnm{Rulkov}, \binits{N.F.}},
\bauthor{\bsnm{Sushchik}, \binits{M.M.}},
\bauthor{\bsnm{Tsimring}, \binits{L.S.}},
\bauthor{\bsnm{Abarbanel}, \binits{H.D.}}:
\batitle{Generalized synchronization of chaos in directionally coupled chaotic
  systems}.
\bjtitle{Phys. Rev. E Stat. Phys. Plasmas Fluids Relat. Interdiscip. Topics}
\bvolume{51}(\bissue{2}),
\bfpage{980}--\blpage{994}
(\byear{1995})
\end{barticle}
\endbibitem

\bibitem[\protect\citeauthoryear{Kocarev and Parlitz}{1996}]{Kocarev1996-mt}
\begin{barticle}
\bauthor{\bsnm{Kocarev}, \binits{L.}},
\bauthor{\bsnm{Parlitz}, \binits{U.}}:
\batitle{Generalized synchronization, predictability, and equivalence of
  unidirectionally coupled dynamical systems}.
\bjtitle{Phys. Rev. Lett.}
\bvolume{76}(\bissue{11}),
\bfpage{1816}--\blpage{1819}
(\byear{1996})
\end{barticle}
\endbibitem

\bibitem[\protect\citeauthoryear{Hunt et~al.}{1997}]{Hunt1997-aw}
\begin{barticle}
\bauthor{\bsnm{Hunt}, \binits{B.R.}},
\bauthor{\bsnm{Ott}, \binits{E.}},
\bauthor{\bsnm{Yorke}, \binits{J.A.}}:
\batitle{Differentiable generalized synchronization of chaos}.
\bjtitle{Phys. Rev. E Stat. Phys. Plasmas Fluids Relat. Interdiscip. Topics}
\bvolume{55}(\bissue{4}),
\bfpage{4029}--\blpage{4034}
(\byear{1997})
\end{barticle}
\endbibitem

\bibitem[\protect\citeauthoryear{Kuptsov and Parlitz}{2012}]{Kuptsov2012-jj}
\begin{barticle}
\bauthor{\bsnm{Kuptsov}, \binits{P.V.}},
\bauthor{\bsnm{Parlitz}, \binits{U.}}:
\batitle{Theory and computation of covariant lyapunov vectors}.
\bjtitle{J. Nonlinear Sci.}
\bvolume{22}(\bissue{5}),
\bfpage{727}--\blpage{762}
(\byear{2012})
\end{barticle}
\endbibitem

\bibitem[\protect\citeauthoryear{Büsing et~al.}{2010}]{Busing2010-mq}
\begin{barticle}
\bauthor{\bsnm{Büsing}, \binits{L.}},
\bauthor{\bsnm{Schrauwen}, \binits{B.}},
\bauthor{\bsnm{Legenstein}, \binits{R.}}:
\batitle{Connectivity, dynamics, and memory in reservoir computing with binary
  and analog neurons}.
\bjtitle{Neural Comput.}
\bvolume{22}(\bissue{5}),
\bfpage{1272}--\blpage{1311}
(\byear{2010})
\end{barticle}
\endbibitem

\bibitem[\protect\citeauthoryear{Jaeger et~al.}{2007}]{Jaeger2007-kc}
\begin{barticle}
\bauthor{\bsnm{Jaeger}, \binits{H.}},
\bauthor{\bsnm{Lukosevicius}, \binits{M.}},
\bauthor{\bsnm{Popovici}, \binits{D.}},
\bauthor{\bsnm{Siewert}, \binits{U.}}:
\batitle{Optimization and applications of echo state networks with
  leaky-integrator neurons}.
\bjtitle{Neural Netw.}
\bvolume{20}(\bissue{3}),
\bfpage{335}--\blpage{352}
(\byear{2007})
\end{barticle}
\endbibitem

\end{thebibliography}
\end{document}